\documentclass[onecolumn,a4paper]{IEEEtran}
\usepackage[utf8]{inputenc} 
\usepackage[T1]{fontenc}
\usepackage{url}              
\usepackage{cite}             
\usepackage[cmex10]{amsmath}
\usepackage{mleftright}       
\mleftright
\usepackage{graphicx}              
\usepackage{booktabs}
\usepackage{hyperref}
\usepackage{cleveref}
\usepackage{amsmath,amssymb,amsthm}
\usepackage{tikz,enumerate}
\usepackage{cite,mathtools,centernot}
\usepackage{mathrsfs,proba,hyperref,cleveref}
\usepackage{amsmath,amssymb,amsthm}
\usepackage{amsmath,amssymb}
\usepackage{tikz,enumerate}
\usepackage{cite,mathtools}
\usepackage{etex}
\theoremstyle{plain}

\newtheorem{theorem}{Theorem} 
\newtheorem{proposition}{Proposition} 
\newtheorem{corollary}{Corollary}

\newtheorem{lemma}{Lemma}

\theoremstyle{remark}

\newtheorem{remark}{Remark}

\theoremstyle{definition}
\newtheorem{definition}{Definition}

\DeclareMathOperator*{\argmin}{arg\,min}

\newcommand{\Mc}{\mathcal{M}}
\newcommand{\Nc}{\mathcal{N}}

\newcommand{\Xc}{\mathcal{X}}
\newcommand{\Yc}{\mathcal{Y}}

\newcommand{\Ah}{{\hat{A}}}

\newcommand{\gh}{{\hat{g}}}

\newcommand{\fh}{{\hat{f}}}

\def\a{\alpha}
\def\b{\beta}
\def\g{\gamma}
\def\d{\delta}
\def\e{\epsilon}
\def\eps{\epsilon}
\def\la{\lambda}
\def\td{\tilde}
\def\o{\omega}
\def\s{\sigma}
\def\ce{\mathfrak{C}}

\def\tf{\mathbb{T}}
\def\gf{\mathbb{G}}
\def\gfa{(\gf,+)}
\def\tfa{(\tf,+)}

\begin{document}

\title{Information Inequalities via Ideas from Additive Combinatorics}

\author{
    \IEEEauthorblockN{Chin Wa (Ken) Lau and Chandra Nair} \\
    \IEEEauthorblockA{
        Department of Information Engineering\\
        The Chinese University of Hong Kong \\
        Sha Tin, NT, Hong Kong \\
        \{kenlau,chandra\}@ie.cuhk.edu.hk
    }
}

\maketitle

\begin{abstract}
    Ruzsa's equivalence theorem provided a framework for converting certain families of inequalities in additive combinatorics to entropic inequalities (which sometimes did not possess stand-alone entropic proofs). In this work, we first establish formal equivalences between some families (different from Ruzsa) of inequalities in additive combinatorics and entropic ones. As a first step to further these equivalences, we establish an information-theoretic characterization of the magnification ratio that could also be of independent interest.
\end{abstract}

\section{Introduction}

\subsection{Background}
We seek to investigate and build upon the analogies and equivalence theorems between sumset inequalities in additive combinatorics and entropic inequalities in information theory. We are directly motivated by the work of Ruzsa \cite{ruz09}, where a formal equivalence theory was proposed and established for certain families of sumset inequalities. Ruzsa categorized the inequalities into three scenarios \cite{ruz09}:
\begin{enumerate}[{Scenario} $A$:]
    \item There exists an equivalence form (see Theorem \ref{thm:ruzsa-equiv}) and explicit implication between a combinatorial inequality and an associated entropic inequality.
    \item A structural analog exists between combinatorial and entropic inequalities, but no direct equivalence is known. Sometimes, one-directional implication can be established.
    \item There is a combinatorial/entropic inequality, but the correctness of counterpart (analogous) inequality is unknown.
\end{enumerate}

Most of the subsequent work has been done along the lines of the second scenario, i.e., analogous entropic inequalities without a formal equivalence. Tao established entropic analogs of the Pl\"unnecke-Ruzsa-Frieman sumset and inverse sumset theory in 2010 \cite{tao10}. Madiman, Marcus, and Tetali established some entropic analogs and equivalence theorems based on partition-determined functions of random variables in 2012\cite{mmt12}. Furthermore, Kontoyiannis and Madiman explored the connection between sumsets and differential entropies \cite{kom14}. We refer readers to \cite{mad08}, \cite{lap08}, \cite{kom10} for more details. One can also find a summary of the connection between combinatorial and entropic inequalities in \cite{dia16}.

The main contributions of this work are the following:
\begin{enumerate}[$a)$]
    \item we establish a formal equivalence theorem (Theorem \ref{thm:gen-rus-type-equiv}) between some combinatorial inequalities and entropic inequalities that Ruzsa had classified into Scenarios B and C. The entropic inequalities take a slightly different form than the analogous ones studied earlier. In some cases, the analogous entropic inequalities are stronger (Remark \ref{rem:stronger-ent}) while in some other cases, the analogous (sometimes conjectured) ones do not imply the equivalent entropic inequalities (Remark \ref{rem:norelation}).
    \item we use information-theory-based arguments to establish some entropic inequalities established by Ruzsa in Scenario A, and some other analogous ones. An entropic equality (Lemma \ref{lem:copy-lemma}), motivated by an analogous combinatorial lemma, has proven repeatedly useful in our arguments.
    \item we prove an information-theoretic characterization of the magnification ratio (Theorem \ref{th:ent-magnification}), which serves as a primitive (as evidenced from Ruzsa's lecture notes \cite{ruz09b}) to larger families of sumset inequalities.
\end{enumerate}

An independent motivation for this study stems from an attempt to prove the subadditivity of certain entropic functionals related to establishing capacity regions in network information theory, the primary research interest of the authors. In one particular but fundamental instance - the Gaussian interference channel - it appears that the additive structure of the channel should play a key role in the proof of the requisite sub-additive inequality.

\subsection{Notation}
We will use $\gfa$ to denote an abelian group and $\tfa$ to denote a finitely generated torsion-free abelian group. If $A$ is a finite set, we use $|A|$ to denote the cardinality of $A$.

\section{Main}

\subsection{Equivalence between sumset inequalities and entropic inequalities}
\label{subsec:ent-equiv}

We state a simple fact below.
There is a trivial equivalence between cardinality inequalities and entropy inequalities via the observation that $\log |A+B| = \max_{P_{XY}} H(X+Y)$, where $X$ takes values in $A$ and $Y$ takes values in $B$. The equality is clearly obtained by taking a uniform distribution on the support of $|A+B|$. However, clearly,  we are seeking non-trivial versions of equivalence theorems.

\begin{theorem}{(Generalized Ruzsa-type Equivalence Theorem)}
    \label{thm:gen-rus-type-equiv}

    Let $\tfa$ be a finitely generated torsion-free abelian group.
    Let $f_1, \dots, f_k$ and $g_1, \dots, g_\ell$ be linear functions on $\tf^n$ with integer coefficients, and let $\alpha_1, \dots, \alpha_k, \beta_1, \dots, \beta_\ell$ be positive real numbers. For the linear function $f_i$, let $S_i\subseteq [1:n]$ denote the index set of non-zero coefficients. Similarly, for $g_i$ let $T_i\subseteq [1:n]$ denote the corresponding index set of non-zero coefficients. (So, effectively, $f_i$ and $g_i$ are linear functions on $\tf_{S_i}$ and $\tf_{T_i}$ respectively). Further, let us assume that  $S_i$ is a pairwise disjoint collection of sets. The following statements are equivalent:
    \begin{enumerate}[$a)$]
        \item For any $A_1, A_2, \dots, A_n$ that are finite subsets of $\tf$, we have
              \begin{align*}
                  \prod_{i = 1}^k |f_i(A_{S_i})|^{\alpha_i} \leq \prod_{i = 1}^\ell |g_i(A_{T_i})|^{\beta_i},
              \end{align*}
              where $A_S = \otimes_{i\in S} A_i$.
        \item For any $m \in \mathbb{N}$, and for any $\hat{A}_1, \Ah_2, \dots, \Ah_n$ that are finite subsets of $\tf^m$, we have
              \begin{align*}
                  \prod_{i = 1}^k |\fh_i(\Ah_{S_i})|^{\alpha_i} \leq \prod_{i = 1}^\ell |\gh_i(\Ah_{T_i})|^{\beta_i},
              \end{align*}
              where $\Ah_S = \otimes_{i\in S} \Ah_i$, and $\fh_i$ (and $\gh_i$) are the natural coordinate-wise extensions of $f_i$ (and $g_i$) respectively, mapping points in $$\underbrace{\tf^m \times \tf^m \times \cdots \times \tf^m}_{n \rm\ times} \mapsto \tf^m.$$
        \item
              For every sequence of random variables $(X_1, \dots, X_n)$, with fixed marginals $P_{X_i}$ and having finite support in  $\tf$, we have
              \begin{align*}
                  \sum_{i = 1}^k \alpha_i \max_{\Pi(X_{S_i})} H(f_i(X_{S_i})) \leq \sum_{i = 1}^\ell \beta_i \max_{\Pi(X_{T_i})} H(g_i(X_{T_i})),
              \end{align*}
              where $\Pi(X_S)$ is collection of joint distributions $P_{X_S}$ that are consistent with the marginals $P_{X_i}, i \in S$.
    \end{enumerate}

\end{theorem}
\begin{proof}
    We will show that $a) \implies b)$, $b) \implies c)$, and $c) \implies a)$. We make a brief remark on the three implications. That $a) \implies b)$ Ruzsa has essentially established in \cite{ruz09} and this is where the requirements that the functions be linear and that the ambient group is finitely generated and torsion-free play a crucial role. Now $b) \implies c)$ is a rather standard argument in the information theory community using the method of types (see Chapter 2 of \cite{csk11}) and Sanov's theorem. Finally $c) \implies a)$ is immediate by taking specific marginal distributions that induce uniform distributions on the support of $f_i(X_{S_i})$ and is where the requirement that $S_i$ be pairwise disjoint plays a role.

    $a) \implies b)$: We outline the method used by Ruzsa in \cite{ruz09}. By the classification theorem of finitely generated abelian groups, we know that a torsion-free finitely generated abelian group is isomorphic to $\mathbb{Z}^d$, for a finite $d$. We denote $t$ to be a generic element in $\tf$, (or equivalently $\mathbb{Z}^d$). Let a linear function with integer coefficients $f: \tf^n \mapsto \tf$, be defined by $f(t_1, \dots,t_n) = \sum_{i=1}^n a_i t_i$. (In the context of our discussion, the locations of the non-zero values of $a_i$ determine the support of $f$). Similarly we denote $\mathbf{t}=(t_1,\dots,t_m)$ to be a generic element in $\tf^m$. Therefore, we have
    $\hat{f}(\mathbf{t}_1,\dots,\mathbf{t}_n) = \sum_{i=1}^n a_i \mathbf{t}_i$.
    Let $\psi_q$ be a linear mapping from $\tf^m$ to $\tf$ defined as
    \begin{align*}
        \psi_q(\mathbf{t}) := t_1 + t_2 q + \dots + t_m q^{m - 1}.
    \end{align*}
    Observe that, by linearity,
    \begin{align}
        \psi_q(\hat{f}(\mathbf{t}_1,..,\mathbf{t}_n)) = \psi_q\left(\sum_{i=1}^n a_i \mathbf{t}_i\right) = f(\psi_q(\mathbf{t}_1),...,\psi_q(\mathbf{t}_n)).
        \label{eq:linear}
    \end{align}
    Given the finite subsets $\hat{A}_1, \dots, \hat{A}_n$ of $\tf^m$, and the linear functions $f_1, \dots, f_k$ and $g_1, \dots, g_\ell$, we can choose a $q$ large enough that $\psi_q(\fh_i(\Ah_{S_i}))$ and $\psi_q(\gh_i(\Ah_{T_i}))$ are injections. Now set $A_i = \psi_q(\hat{A}_i)$.  Therefore we have
    \begin{align*}
        | \fh_i(\Ah_{S_i}) | = |\psi_q(\fh_i(\Ah_{S_i}))| \stackrel{(a)}{=} |f_i(\{\psi_q(\Ah_k)\}_{k \in S_i})| = | f_i (A_{S_i})|,
    \end{align*}
    where $(a)$ follows from \eqref{eq:linear}. A similar equality holds for $g$'s as well. With these equalities, we have that $a) \implies b)$.

    $b) \implies c)$: We are given a set of marginal distributions $P_{X_1}, \dots, P_{X_n}$ whose supports are finite subsets of $\tf$, say $\Xc_1,\dots,\Xc_n$. Consider a non-negative sequence $\{\delta_m\}$, where $\delta_m \rightarrow 0$ and $\sqrt{m} \cdot \delta_m \rightarrow \infty$ as $m \rightarrow \infty$. For every $m$, we construct the strongly typical sets   $\mathsf{T}_{(m,P_{X_i},\delta_m)}$, for  $1 \leq i \leq n$, where
    \begin{align*}
         & \mathsf{T}_{(m,P_{X_i},{\delta_m})}:= \bigg\{\mathbf{x} \in \mathcal{X}_i^m: \left|\frac{1}{m} N(a | \mathbf{x}) - P_{X_i}(a) \right|  \leq \delta_m \cdot P_{X_i}(a) \text{ for any $a \in \mathcal{X}_i$} \bigg\}.
    \end{align*}
    Suppressing dependence on other variables, let $\hat{A}_i = \mathsf{T}_{(m,P_{X_i},{\delta_m})}$ for $1 \leq i \leq n$. Now consider a linear function $f: \tf_S \mapsto \tf$ and let $\hat{f}$ be the coordinate-wise extension of it to $(\tf^m)_S$. Define $Y=f(X_S)$, $S \subseteq [1:n]$, and let $\mathcal{M}_{Y}$ denote the set of probability distributions of $Y$ induced by all couplings $\Pi(X_S)$ that are consistent with the marginals $P_{X_i}$ for $i \in S$. Let $Q_Y$ be the uniform distribution on $\Yc$, and by a routine application\footnote{This is standard in certain information theory circles. For completeness, we outline a proof of the maximum coupling by the discrete Sanov theorem  in Appendix \ref{sec:discrete-sanov} and Appendix \ref{sec:max-coup}.} of Sanov's theorem we obtain that
    \begin{align*}
         & \lim_{m\to \infty} \frac{1}{m} \log \frac{|\hat{f}(\Ah_{S})|}{|\Yc|^{m}} = \max_{P_Y \in \mathcal{M}_{Y}} H(Y) - \log |\Yc|  = \max_{\Pi(X_S)} H(f(X_S)) -  \log |\Yc|.
    \end{align*}
    Therefore, we have
    \begin{align*}
        \lim_{m\to \infty} \frac{1}{m} |\hat{f}(\Ah_{S})| = \max_{\Pi(X_{S})} H(f(X_{S})).
    \end{align*}
    Thus, the implication $b) \implies c)$ is established. 

    $c) \implies a)$: This is rather immediate. Since $S_i$'s are pairwise disjoint, let $P_{X_{S_i}}$ induce a uniform distribution on $f(A_{S_i})$ and let $P_{X_i}$ be the induced marginals. Then it is clear that $\max_{\Pi(X_{S_i})} H(f_i(X_{S_i})) = \log|f(A_{S_i})|$ and $\max_{\Pi(X_{T_i})} H(g_i(X_{T_i})) \leq \log|g(A_{T_i})|$ and this completes the proof.
\end{proof}

The following corollaries to Theorem \ref{thm:gen-rus-type-equiv} lead to some entropic inequalities. Some of the sumset inequalities in literature are stated using Ruzsa-distance, and the equivalent entropic inequalities can be stated using a similar distance between distributions.

\begin{definition}{(Ruzsa Distance between Finite Sets, \cite{ruz96})}
    \label{def:rus-dist}
    The Ruzsa distance between two finite subsets $A, B$ on an abelian group $\gfa$ is defined as
    \begin{align*}
        d_R(A, B) := \log \frac{|A - B|}{|A|^{1/2} |B|^{1/2}}.
    \end{align*}
\end{definition}
\begin{remark}
    It is clear that $d_R(A,B) = d_R(B,A)$ and that $d_R(A,A) \geq 0$.
\end{remark}

\begin{definition}{(Entropic Ruzsa Distance)}
    \label{def:ent-rus-dist}
    The entropic-Ruzsa ``distance'' between two distributions $P_X, P_Y$ taking values in $\gfa$ is defined as
    \begin{align*}
        d_{HR}(X, Y) := \max_{P_{XY} \in \Pi(P_X, P_Y)} H(X - Y) - \frac{1}{2} H(X) - \frac{1}{2} H(Y),
    \end{align*}
    where $\Pi(P_X, P_Y)$ is the set of all coupling with the given marginals.
\end{definition}
\begin{remark}
    \label{rem:ruz-dis-ent}
    The following remarks are worth noting with regard to the entropic Ruzsa-distance:
    \begin{enumerate}
        \item As with the abuse of notations in information theory $d_{HR}(X, Y)$ is a function of  $P_X, P_Y$ and not of $X$ and $Y$.
        \item Just like the original Ruzsa distance between two sets, we have $d_{HR}(X, Y) \geq 0$ (this follows by observing that when $P_{XY}=P_X P_Y$, we have $H(X-Y) \geq \max\{H(X),H(Y)\}$ as $0 \leq I(X;X-Y) = H(X-Y) - H(Y)$). Further it is immediate that $d_{HR}(X, Y) = d_{HR}(Y, X)$.
        \item There is no ordering between $d_{HR}(X, Y)$ and $d_R(A,B)$ where $A$ is the support of $p_X$ and $B$ is the support of $p_Y$.
              \begin{itemize}
                  \item Consider $P_X$ and $P_Y$ such that it is uniform on sets $A$ and $B$ respectively. Thus for any $P_{XY} \in \Pi(P_X, P_Y)$ we have $H(X-Y) \leq \log|A-B|$ and consequently $d_{HR}(X, Y) \leq d_R(A,B)$ (and the inequality can be strict).
                  \item Consider a joint $P_{XY}$ that is uniform on $A-B$ and let $P_X$ and $P_Y$ be its induced marginal distributions on sets $A$ and $B$ respectively. then as $H(X) \leq \log|A|$ and $H(Y) \leq \log|B|$, we have $d_{HR}(X,Y) \geq d_R(A,B)$ (and the inequality can be strict).
              \end{itemize}
        \item This definition is different from that of Tao \cite{tao10}, where he defines the similar quantity using independent coupling of $P_X$ and $P_Y$. An advantage of our definition is that we have a formal equivalence between the two inequalities (one in sumset and one in entropy). Independent of this work, in \cite{gmt23} the authors also defined the same notion of distance (see the section on bibliographic remarks) and called it the \textit{maximal entropic Ruzsa's distance}.
    \end{enumerate}
\end{remark}

Theorem \ref{thm:gen-rus-type-equiv} immediately implies the following entropic inequalities from the corresponding sumset inequalities.
\begin{corollary}
    For any distributions $P_X, P_Y, P_Z$ with finite support on a finitely generated torsion-free group $\tfa$, we have
    \begin{align}
        d_{HR}(X, Z)            & \leq d_{HR}(X, Y) + d_{HR}(Y, Z), \nonumber                                                                                   \\
        \mbox{or equivalently}: & ~ H(Y) + \max_{\Pi(X, Z)} H(X - Z) \leq  \max_{\Pi(X, Y)} H(X - Y) +  \max_{\Pi(Y, Z)} H(Y - Z). \label{eq:sum-diff-ent-ineq}
    \end{align}
\end{corollary}

\begin{proof}
    In \cite{ruz96}, Ruzsa showed that for any finite $A, B, C$ on a finitely generated torsion-free abelian group $\tfa$, we have $d_R(A, C) \leq d_R(A, B) + d_R(B, C)$, or equivalently $|B||A-C| \leq |A-B||B-C|$.
    We will obtain the desired inequality by applying Theorem \ref{thm:gen-rus-type-equiv}.
    \begin{remark}
        \label{rem:stronger-ent}
        The entropic inequality in \eqref{eq:sum-diff-ent-ineq} can also be obtained as a direct consequence of a stronger entropic inequality that was established in \cite{mmt12}. There, it was established that if $Y$ and $(X, Z)$ are independent and taking values in an ambient abelian group $\gfa$, then one has $H(Y) + H(X - Z) \leq H(X - Y) + H(Y - Z)$. To see this, observe that $H(Y,X-Z) = H(X-Y,Y-Z) - I(X;Y -Z|X-Z)$, and the requisite inequality is immediate.
    \end{remark}
\end{proof}

\begin{corollary}
    \label{cor:no-ent-proof}
    For distributions $P_X, P_Y, P_Z$ with finite support on a finitely generated torsion-free group $\tfa$, we have
    \begin{align}
        \label{eqn:sumset-ent-equiv}
        \begin{split}
            &H(X) + \max_{\Pi(Y, Z)} H(Y + Z) \leq \max_{\Pi(X, Y)} H(X + Y) + \max_{\Pi(X, Z)} H(X + Z).
        \end{split}
    \end{align}
\end{corollary}

\begin{proof}
    In \cite{ruz96}, Ruzsa showed that for any finite $A, B, C$ on a finitely generated torsion-free abelian group $\tfa$, we have
    \begin{align}
        \label{eqn:sumset}
        |A| |B + C| \leq |A + B| |A + C|.
    \end{align}
    We obtain the desired entropic inequality by applying Theorem \ref{thm:gen-rus-type-equiv}.
\end{proof}
\begin{remark}
    The authors are unaware of a stand-alone information-theoretic proof of the above inequality. Our results in Section \ref{sec:magnification} is a step at building an information-theoretic counterpart to the sumset arguments used to establish this. When $X, Y,$ and $Z$ are mutually independent, an entropic analog has been established in \cite{mad08,mmt12}. Note that in this case, by the data-processing inequality, we have
    $I(Z;X+Y+Z) \leq I(Z;X+Z)$ implying
    \[H(X) + H(Y + Z) \leq H(X) + H(X + Y + Z)  \leq H(X + Y) + H(X + Z) \]
    A relaxation of this proof to the case, when $X$ is independent of $(Y,Z)$, would have yielded \eqref{eqn:sumset-ent-equiv}; however, this relaxation does not seem immediate.
\end{remark}

\subsection{Katz-Tao Sum Difference Inequality}
\label{subsec:katz-tao}
Katz and Tao established the following lemma \cite{kat99} and used it in the proof of certain sumset inequalities.
\begin{lemma}{\cite[Lemma 2.1]{kat99}}
    \label{lem:combinatorial-copy-lemma}
    Let $A$ and $B_1, \dots, B_{n - 1}$ be finite sets for some $n$. Let $f_i : A \rightarrow B_i$ be a function for all $i \in [1 : n - 1]$. Then
    \begin{align*}
        \{(a_1, \dots, a_n) \in A^n : f_i(a_i) = f_i(a_{i + 1}) \forall \text{  $i \in [1 : n - 1]$}\} \geq \frac{|A|^n}{\prod_{i = 1}^{n - 1} |B_i|}.
    \end{align*}
\end{lemma}

Motivated by this lemma, we will prove an information-theoretic version (which would imply the combinatorial version) and will turn out to be useful in several of our arguments. We will first present a lemma in a more general form than is used in this paper.

\begin{lemma}
    \label{lem:longmclem}
    Suppose the following Markov chain holds:
    \begin{align*}
        X_1 \rightarrow U_1 \rightarrow X_2 \rightarrow U_2 \rightarrow \dots \rightarrow X_{n - 1} \rightarrow U_{n - 1} \rightarrow X_n.
    \end{align*}
    Then,
    \begin{align*}
        H(X_1, \dots, X_n, U_1, \dots, U_{n - 1}) + \sum_{i = 1}^{n - 1} I(X_i; U_i) + \sum_{i = 1}^{n - 1} I(U_i; X_{i + 1}) = \sum_{i = 1}^n H(X_i) + \sum_{i = 1}^{n - 1} H(U_i).
    \end{align*}
\end{lemma}
\begin{proof}
    This lemma is an immediate consequence of the Chain Rule for entropy as follows. Note that the chain rule and the Markov Chain assumption yields
    \begin{align*}
        H(X_1, \dots, X_n, U_1, \dots, U_{n - 1}) & = H(X_1) + \sum_{i=1}^{n-1} H(U_i|X_i) + \sum_{i=1}^{n-1} H(X_{i+1}|U_i)                                                 \\
                                                  & = H(X_1) + \sum_{i=1}^{n-1} \big( H(U_i) - I(U_i;X_i)\big) + \sum_{i=1}^{n-1} \big( H(X_{i+1}) - I(U_i; X_{i + 1})\big).
    \end{align*}
    Now, rearranging yields the desired equality.
\end{proof}

As a special case of \Cref{lem:longmclem} we obtain the following version that is useful in this paper.

\begin{lemma}
    \label{lem:copy-lemma}
    Let $(X_i)_{i = 1}^n$ be a sequence of finite-valued random variables (defined on some common probability space)  and $(f_i, g_i)_{i = 1}^{n - 1}$ be a sequence of functions that take a finite set of values in some space $\mathcal{S}$ such that:   $f_i(X_i) = g_i(X_{i + 1}) (=: U_i)$ and the following Markov chain holds,
    \begin{align*}
        X_1 \rightarrow U_1 \rightarrow X_2 \rightarrow U_2 \rightarrow \dots \rightarrow X_{n - 1} \rightarrow U_{n - 1} \rightarrow X_n.
    \end{align*}
    Then,
    \begin{align*}
        H(X_1, \dots, X_n) + \sum_{i = 1}^{n - 1} H(U_i) = \sum_{i = 1}^n H(X_i).
    \end{align*}
\end{lemma}
\begin{proof}
    Note that
    $H(X_1, \dots, X_n)=H(X_1, \dots, X_n, U_1,\dots, \allowbreak U_{n-1})$ since $U_i$ is determined by $X_i$ (and also by $X_{i+1}$). Further we also have $I(U_i;X_i) = I(U_i;X_{i+1})=H(U_i)$ for $1 \leq i \leq n-1$. Hence, the desired consequence follows from Lemma \ref{lem:longmclem}.
\end{proof}
\begin{remark} The following remarks are worth noting:
    \begin{itemize}
        \item Lemma \ref{lem:copy-lemma} seems to play a similar role as the copy lemma \cite{Zhang1998} used in deriving several non-Shannon type inequalities.
        \item Note that Lemma \ref{lem:copy-lemma} will imply Lemma \ref{lem:combinatorial-copy-lemma} directly. Define
              $$C = \{(a_1, \dots, a_n) \in A^n : f_i(a_i) = f_i(a_{i + 1}) \forall \text{  $i \in [1 : n - 1]$}\}.$$
              Suppose $X_1, \dots, X_n$ have uniform marginals on $A$.  Set $f_i(X_i) = f_i(X_{i + 1}) (=: U_i)$ and construct a joint distribution such that  the following Markov chain holds,
              \begin{align*}
                  X_1 \rightarrow U_1 \rightarrow X_2 \rightarrow U_2 \rightarrow \dots \rightarrow X_{n - 1} \rightarrow U_{n - 1} \rightarrow X_n.
              \end{align*}
              Observe that $(X_1, \dots, X_n)$ has a support on $C$. This implies, from Lemma \ref{lem:copy-lemma}, that
              \begin{align*}
                  n \log |A| & = \sum_{i = 1}^n H(X_i) = H(X_1, \dots, X_n) + \sum_{i = 1}^{n - 1} H(U_i) \\
                             & \quad \leq \log |C| + \sum_{i = 1}^{n - 1} \log |B_i|.
              \end{align*}
    \end{itemize}
\end{remark}

The main intent of the remainder of the section is to demonstrate the role of Lemma \ref{lem:copy-lemma} to establish various information inequalities.

\begin{definition}{($G$-restricted Sumset \cite{ruz09})}
    Suppose $G$ is a subset of $A \times B$, where $A, B$ are finite subsets of $\gfa$. We denote the $G$-restricted sumset and difference set of $A$ and $B$ as $A \overset{G}{+} B$ and $A \overset{G}{-} B$.
    \begin{align*}
        A \overset{G}{+} B & = \{a + b : a \in A, b \in B, (a, b) \in G\}, \\
        A \overset{G}{-} B & = \{a - b : a \in A, b \in B, (a, b) \in G\}.
    \end{align*}
\end{definition}

\begin{theorem}{(Katz-Tao Sum-Difference Inequality \cite{kat99})}
    \label{thm:katz-tao-comb}
    
    For any $G$, a finite subset of $\tf \times \tf$, we have
    \begin{align*}
        |A \overset{G}{-} B| \leq |A|^{2/3} |B|^{2/3} |A \overset{G}{+} B|^{1/2}.
    \end{align*}
\end{theorem}

In \cite{ruz09}, Ruzsa established the entropy version of Katz-Tao sum-difference inequality by using a formal equivalence theorem between G-restricted sumset inequalities and entropic inequalities.

\begin{theorem}[Equivalence Theorem 2, \cite{ruz09}]
    \label{thm:ruzsa-equiv}

    Let $f, g_1, \dots, g_k$ be linear functions in two variables with integer coefficients, and let $\alpha_1, \dots, \alpha_k$ be positive real numbers. The following statements are
    equivalent:
    \begin{enumerate}
        \item For every finite $A \subseteq \tf \times \tf$ we have
              \begin{align*}
                  |f(A)| \leq \prod |g_i(A)|^{\alpha_i},
              \end{align*}
              where $|f(A)|$ denotes the cardinality of the image $f(A)$.
        \item
              For every pair $X, Y$ of (not necessarily independent) random variables with values in $\tfa$ such that the entropy of each $g(X, Y)$ is finite, the entropy
              of $f(X, Y)$ is also finite and it satisfies
              \begin{align*}
                  H(f(X, Y)) \leq \sum \alpha_i H(g_i(X, Y)).
              \end{align*}
    \end{enumerate}
\end{theorem}

\begin{remark}
    It may be worthwhile mentioning a key difference between Theorem \ref{thm:gen-rus-type-equiv} and \Cref{thm:ruzsa-equiv}. The equivalence in Theorem \ref{thm:ruzsa-equiv} follows when the sum-set inequalities hold for every $G$-restricted sumset. On the other hand, most of the inequalities in literature are established for the Minkowski sum of sets, and Theorem \ref{thm:gen-rus-type-equiv} holds under such a situation.
\end{remark}

Consequently, Ruzsa obtained the following entropic inequality by applying  Theorem \ref{thm:ruzsa-equiv} to \Cref{thm:katz-tao-comb}.
\begin{theorem}{\cite{ruz09}}
\label{thm:katz-tao-entropic}
    Suppose $X$ and $Y$ are random variables with finite support on $\tfa$, we have
    \begin{align}
        \label{eqn:kat-entropy}
        H(X - Y) \leq \frac{2}{3} H(X) + \frac{2}{3} H(Y) + \frac{1}{2} H(X + Y).
    \end{align}
\end{theorem}

\begin{theorem} \cite[Proposition 3.6]{tav06ds}
    \label{th:katz-tao-ent}
    Suppose $X$ and $Y$ are random variables with finite support on an ambient abelian group $\mathbb{G}$, we have
    \begin{align*}
        \frac{1}{2} I(X; X - Y) + \frac{1}{2} I(Y; X - Y) \leq  \frac{3}{2} I(X; X + Y) + \frac{3}{2} I(Y; X + Y) + 3 I(X; Y).
    \end{align*}
\end{theorem}
The proof of this Theorem will be presented in the Appendix for the sake of completeness. The only minor difference between the arguments is using Lemma \ref{lem:longmclem} instead of the submodularity argument used in \cite{tav06ds}.

\begin{remark}
    The following remarks (and acknowledgments) may interest a careful reader.
    \begin{itemize}
    \item The authors were originally unaware of an  entropic argument by Tao and Vu (found in \cite{tav06ds}) for the result in Theorem \ref{th:katz-tao-ent}. However, the authors were notified of the same by Prof. Ben Green shortly upon uploading a preliminary version of this to arXiv.
        \item Since the quantities in Theorem 4 are mutual information and not entropies, it is immediate that the same inequality also holds from continuous-valued random variables or those whose support is not finite.
        \item The version of the result that was presented at ISIT in July 2023 was of the form presented in Corollary \ref{co:katz-tao-ruzsa}. Subsequently, Lampros Gavalakis and Ioannis Kontoyannis wondered about a continuous random variable analog of Corollary \ref{co:katz-tao-ruzsa}. Following this discussion, we modified our original proof to develop Theorem \ref{th:katz-tao-ent}. Later, we realized that Tao and Vu \cite{tav06ds} already had this result in an unpublished manuscript.
    \end{itemize}
\end{remark}
\begin{corollary}
    Suppose $X$ and $Y$ are random variables with finite support on an ambient abelian group $\gf$. We have
    \begin{align}
        \label{eqn:kat-entropy-mod}
        H(X - Y) \leq \frac{2}{3} H(X) + \frac{2}{3} H(Y) + \frac{1}{2} H(X + Y).
    \end{align}
    \label{co:katz-tao-ruzsa}
\end{corollary}
\begin{proof}
    Following the equivalent form of the result in \eqref{eq:equivform}, we have
    \begin{align*}
        0 & \geq 5 H(X, Y) - 4 H(X) - 4 H(Y) - 3 H(X + Y) + H(X - Y) \\
          & \geq 6 H(X - Y) - 4 H(X) - 4 H(Y) - 3 H(X + Y).
    \end{align*}

    The second inequality holds if and only if $(X, Y)$ is a function of $X - Y$.
\end{proof}

One also immediately obtain the following independently established results as a corollary of Theorem \ref{th:katz-tao-ent}.

\begin{corollary}
    \label{cor:indep}
    \cite[Theorem 1.10]{tao10} If $X$ and $Y$ are independent discrete-valued random variables
    \begin{align*}
        H(X - Y) \leq  3 H(X + Y) - H(X) - H(Y).
    \end{align*}
    \cite[Theorem 3.7]{kom14} If $X$ and $Y$ are independent continuous-values random variables with well-defined differential entropies
    \begin{align*}
        h(X - Y) \leq 3 h(X + Y) - h(X) - h(Y).
    \end{align*}
\end{corollary}
\begin{proof}
    The independence between $X$ and $Y$ reduces the inequality established in Theorem \ref{th:katz-tao-ent} to
    \begin{align*}
         & \frac{1}{2} I(X; X - Y) + \frac{1}{2} I(Y; X - Y) \leq  \frac{3}{2} I(X; X + Y) + \frac{3}{2} I(Y; X + Y),
    \end{align*}
    which is equivalent to
    \begin{align*}
        H(X-Y) & \leq 3 H(X+Y) - \frac{3}{2}H(X+Y|X)  - \frac{3}{2}H(X+Y|Y) + \frac{1}{2}H(X-Y|X) + \frac{1}{2}H(X-Y|Y) \\
               & = 3 H(X+Y) - H(X) - H(Y).
    \end{align*}
    The proof for the continuous case is identical.
\end{proof}

\begin{theorem}{(Sum-difference Inequality)\cite[Theorem 5.3]{ruz96}}
    \label{thm:sum-diff}
    The Ruzsa distance between two finite subsets $A, B$ on an abelian group $\gfa$ satisfies
    \begin{align}
        \label{eqn:sum-diff}
        \begin{split}
            &d_R(A, -B) \leq 3 d_R(A, B), \\
            & \mbox{or equivalently ~} |A + B| |A| |B| \leq |A-B|^3.
        \end{split}
    \end{align}
\end{theorem}

\begin{remark}
The following proposition and the proof below are essentially identical to that of Proposition 2.4 in \cite{tav06ds}. The only (minor) difference is that we do not assume that $X$ and $Y$ are independent.
\end{remark}

\begin{proposition}{(Entropic Sum-difference Inequality)}
    \label{prop:entropy-sum-diff}
    Let $X_1, Y_1, X_2, Y_2, X_3, Y_3$ be random variables (on a common probability space) with finite support on an abelian group $\gfa$ such that $X_1 - Y_1 = X_2 - Y_2 ~(=:U)$  and also satisfies that $(X_1, Y_1) \rightarrow U \rightarrow (X_2, Y_2)$ forms a Markov chain. Further, suppose $(X_1, Y_1, X_2, Y_2)$ and $(X_3, Y_3)$ are independent. Then the following inequality holds:
    \begin{align}
        \label{eqn:ent-sum-diff}
        \begin{split}
            &H(X_1, Y_1) + H(X_2, Y_2) + H(X_3 + Y_3) \leq H(X_1-Y_1) + H(X_1, Y_2, X_2 - Y_3, X_3 - Y_1).
        \end{split}
    \end{align}
\end{proposition}

\begin{proof}
    Since $U = X_1 - Y_1 = X_2 - Y_2$ and $(X_1, Y_1) \rightarrow U \rightarrow (X_2, Y_2)$ forms a Markov chain, from Lemma \ref{lem:copy-lemma} we have
    \begin{align}
         & H(X_1, Y_1, X_2, Y_2) + H(U)  = H(X_1, Y_1) + H(X_2, Y_2) \label{eq:one}
    \end{align}

    We now decompose $H(X_1, Y_1, X_2, Y_2, X_3, Y_3 | X_3 + Y_3)$ in two ways. Firstly, since $(X_1, Y_1, X_2, Y_2)$ and $(X_3, Y_3)$ are independent, we have
    \begin{align*}
         & H(X_1, Y_1, X_2, Y_2, X_3, Y_3 | X_3 + Y_3)                                  \\
         & = H(X_1, Y_1, X_2, Y_2) + H(X_3, Y_3 | X_3 + Y_3)                            \\
         & \overset{(a)}{=} H(X_1, Y_1) + H(X_2, Y_2) - H(U) + H(X_3, Y_3 | X_3 + Y_3),
    \end{align*}
    where $(a)$ follows by \eqref{eq:one}.

    On the other hand, we have
    \begin{align*}
         & H(X_1, Y_1, X_2, Y_2, X_3, Y_3 | X_3 + Y_3)                                            \\
         & = H(X_1, Y_2, X_2 - Y_3, X_3 - Y_1, X_3, Y_3 | X_3 + Y_3)                              \\
         & \leq H(X_1, Y_2, X_2 - Y_3, X_3 - Y_1 | X_3 + Y_3) + H(X_3, Y_3 | X_3 + Y_3)           \\
         & = H(X_1, Y_2, X_2 - Y_3, X_3 - Y_1 , X_3 + Y_3) - H(X_3+Y_3) + H(X_3, Y_3 | X_3 + Y_3) \\
         & = H(X_1, Y_2, X_2 - Y_3, X_3 - Y_1 ) - H(X_3+Y_3) + H(X_3, Y_3 | X_3 + Y_3).
    \end{align*}
    The last equality is a consequence of the observation that $(X_1, Y_2, X_2 - Y_3, X_3 - Y_1)$ implies $(X_1,Y_2,X_2 + Y_1 -( X_3 + Y_3))$. However as $X_1 + Y_2 = X_2 + Y_1$ by assumption, we observe that $H(X_3+Y_3|X_1, Y_2, X_2 - Y_3, X_3 - Y_1)=0$ and thus justifying the equality.

    By combining these two decompositions, we obtain
    \begin{align*}
         & H(X_1, Y_1) + H(X_2, Y_2) + H(X_3 + Y_3)  \leq H(U) + H(X_1, Y_2, X_2 - Y_3, X_3 - Y_1).
    \end{align*}
\end{proof}

\begin{remark}
    The arguments here are also motivated by similar arguments in the sumset literature \cite{greennotes} and in Tao's work on a similar inequality in \cite{tao10}.
\end{remark}

\begin{corollary}
    \label{cor:sum-diff}
    In addition to the assumptions on $X_1, Y_1, X_2, Y_2, X_3, Y_3$ imposed in Proposition \ref{prop:entropy-sum-diff}, let us assume that $X_1$ is independent of $Y_1$  and  $X_2$ independent of $Y_2$. Then we have
    \begin{align*}
         & H(X_2) + H(Y_1) + H(X_3 + Y_3)  \leq H(X_1 - Y_1) + H(X_3 - Y_1) + H(X_2 - Y_3).
    \end{align*}
\end{corollary}
\begin{proof}
    The proof is immediate from Proposition \ref{prop:entropy-sum-diff} along with the observation that the assumptions imply $H(X_1,Y_1)=H(X_1) + H(Y_1)$, $H(X_2,Y_2)=H(X_2) + H(Y_2)$, and using the sub-additivity of entropy applied to
    $H(X_1,Y_2,X_2-Y_3,X_3-Y_1)$.
\end{proof}
\begin{remark}
    \label{rem:couplingexist}
    Suppose $X$ and $Y$ are independent random variables having finite support on $\gf$, and random variables $X_3,Y_3$ also have finite support on $\gf$, then observe that we can always construct a coupling $(X_1, Y_1, X_2, Y_2, X_3, Y_3)$ satisfying the assumptions of Corollary \ref{cor:sum-diff}, so that $(X_1,Y_1)$ and $(X_2,Y_2)$ are distributed as $(X,Y)$.
\end{remark}

\begin{corollary}[Generalized Ruzsa sum-difference inequality]
    \label{cor:gen-ruz-sum-diff}
    Let $A,B,C,D$ be finite subsets of an abelian group $(\gf,+)$. Then the following sumset inequality holds:
    $$|A||B| |C+D| \leq |A-B| |C-B| |A - D|,$$
    or equivalently
    $$ d_R(C,-D) \leq d_R(C,B) + d_R(B,A) + d_R(A,D). $$
\end{corollary}
\begin{proof}
    Suppose $X$ is a uniform distribution on $A$ and $Y$ is a uniform distribution on $B$. Further let $X_3,Y_3$ be taking values on $C, D$ (respectively) such that $X_3 + Y_3$ is uniform on $C+D$. Let $(X_1,Y_2,X_2,Y_2,X_3,Y_3)$ be the coupling according to \Cref{rem:couplingexist} and observe that Corollary \ref{cor:sum-diff} implies that
    \begin{align*}
         & \log |A| + \log |B| + \log |C+D|                 \\
         & \quad  \leq H(U) + H(X_3-Y_1) + H(X_2-Y_3)       \\
         & \quad \leq \log |A-B| + \log |C-B| + \log |A-D|.
    \end{align*}

    Here, the second inequality used the fact that the entropy of a finite valued random variable is upper bounded by the logarithm of its support size.
\end{proof}
\begin{remark}
    Setting $C = A$ and $D = B$, we can see that the above is a generalization of Theorem \ref{thm:sum-diff}.
\end{remark}

\begin{corollary}
    For any distributions $P_U, P_V, P_X, P_Y$ with finite support on a finitely generated torsion-free group $\tfa$, we have
    \begin{align*}
         & H(X) + H(Y) + \max_{\Pi(U, V)} H(U + V) \leq                                             \\
         & \quad \max_{\Pi(X, Y)} H(X - Y) + \max_{\Pi(X, U)} H(X - U) + \max_{\Pi(V, Y)} H(V - Y).
    \end{align*}

\end{corollary}

\begin{proof}
    From Corollary \ref{cor:sum-diff}, for any finite $A, B, C, D$ on a finitely generated torsion-free abelian group $\tfa$, we have
    \begin{align*}
        |A| |B| |C + D| \leq |A - B| |A - D| |C - B|.
    \end{align*}
    We will obtain the desired inequality by applying Theorem \ref{thm:gen-rus-type-equiv}.
\end{proof}

\begin{remark}
    Setting $U = Y$ and $V = X$ from the above result. We will obtain an entropic analog of sum-difference inequality
    \begin{align*}
         & d_{HR}(X, -Y) \leq 3 d_{HR}(X, Y),                                                                 \\
         & \mbox{or equivalently ~} H(X) + H(Y) + \max_{\Pi(X, Y)} H(X + Y) \leq 3 \max_{\Pi(X, Y)} H(X - Y).
    \end{align*}
\end{remark}

\begin{remark}
    \label{rem:norelation}
    There seems to be no direct implication between these two statements:
    \begin{itemize}
        \item Suppose $X$ and $Y$ are independent, we have $H(X) + H(Y) + H(X + Y) \leq 3 H(X - Y)$.
              This was the previously considered analogous form of the sum-difference inequality \eqref{eqn:sum-diff}, established in \cite{tao10}.
        \item For any $P_X, P_Y$, we have
              \begin{align*}
                  H(X) + H(Y) + \max_{\Pi(X, Y)} H(X + Y) \leq 3 \max_{\Pi(X, Y)} H(X - Y).
              \end{align*}
              This is the formally established equivalent form of the sum-difference inequality \eqref{eqn:sum-diff}.
    \end{itemize}
\end{remark}

\section{Entropic Formulation of Magnification Ratio}
\label{sec:magnification}
There are a large number of sumset inequalities that do not have entropic equivalences yet, such as Plünnecke–Ruzsa inequality (even though some entropic analogs have been established in \cite{tao10, kom14}). A combinatorial primitive that frequently occurs in the combinatorial proofs is the notion of a magnification ratio (see the lecture notes: \cite{ruz09b}).
In this section, we establish an entropic characterization of the magnification ratio, and in addition to this result potentially being useful in deriving new entropic equivalences, it may also be of independent interest to the combinatorics community.

Let $G \subseteq A \times B$ be a finite bipartite graph with no isolated vertices in $A$ or $B$. For every $S \subseteq A$, let $\mathcal{N}(S) \subseteq B$ denote the set of neighbors of $S$.

\begin{definition}
    The magnification ratio of $G$ from $A$ to $B$ is defined as
    \begin{align*}
        \mu_{A \rightarrow B}(G) = \min_{S \subseteq A, S \neq \emptyset} \frac{|\mathcal{N}(S)|}{|S|}.
    \end{align*}
\end{definition}

\begin{definition}{(Channel Consistent with a Bipartite Graph)}
    Let $\mathcal{W}$ be the set of all possible channels (or probability transition matrices) from $A$ to $B$. Given a bipartite graph $G \subseteq A \times B$, we define
    \begin{align*}
        \mathcal{W}(G) := \{W \in \mathcal{W} : W(Y = b | X = a) = 0 \text{ if $(a, b) \notin G$}\},
    \end{align*}
    to be the set of all channels consistent with the bipartite graph $G$. Note that $\mathcal{W}(G)$ is a closed and compact set.
\end{definition}
In the above, we think of $X$ (taking values in $A$) as the input and $Y$ (taking values in $B$) as the output of a channel $W_{Y|X}$.
Given an input distribution $P_X$, we define
\begin{align*}
    \lambda_{A \rightarrow B} (G; P_X) := \max_{W \in \mathcal{W}(G)} (H(Y) - H(X)).
\end{align*}
Given a fixed $P_X$, it is rather immediate that $H(Y)$ is concave in $W_{Y|X}$.
Let $W^*(G; P_X) \in \mathcal{W}(G)$ denote a corresponding optimizer, i.e.
\begin{align*}
    W^*(G; P_X) := \arg \max_{W \in \mathcal{W}(G)} (H(Y) - H(X)).
\end{align*}
If the optimizer is a convex set, we define it to be an arbitrary element of this set.

Finally, we define the quantity
\begin{align}
    \label{eq:ent-mag}
     & \lambda_{A \rightarrow B}(G) := \min_{P_X} \lambda_{A \rightarrow B} (G; P_X) \\
     & \quad = \min_{P_X} \max_{W \in \mathcal{W}(G)} (H(Y) - H(X)).
\end{align}

The main result of this section is the following result.
\begin{theorem}[Entropic characterization of the magnification ratio]
    \label{th:ent-magnification}
    \begin{align*}
         & \log \mu_{A \rightarrow B}(G)  = \lambda_{A \rightarrow B}(G), ~ \mbox{or equivalently}, \\
         & \log \mu_{A \rightarrow B}(G) = \min_{P_X} \max_{W \in \mathcal{W}(G)} (H(Y) - H(X)).
    \end{align*}
\end{theorem}

\begin{proof}
    We first establish that $ \lambda_{A \rightarrow B}(G) \leq \log \mu_{A \rightarrow B}(G).$ This direction is rather immediate. Let
    \begin{align*}
        A^* := \argmin_{S \subseteq A, S \neq \emptyset} \frac{|\mathcal{N}(S)|}{|S|}.
    \end{align*}
    So we have $\mu_{A \rightarrow B}(G) = \frac{|\mathcal{N}(A^*)|}{|A^*|}$.
    Let $P_X$ be the uniform distribution on $A^*$. Then note that
    \begin{align*}
        \lambda_{A \rightarrow B}(G) & \leq \lambda_{A \rightarrow B}(G; P_X)                                      \\
                                     & = \max_{W \in \mathcal{W}(G)} (H(Y) - H(X))                                 \\
                                     & = \max_{W \in \mathcal{W}(G)} (H(Y) - \log |A^*|)                           \\
                                     & \leq \log |\mathcal{N}(A^*)| - \log |A^*| = \log  \mu_{A \rightarrow B}(G).
    \end{align*}
    This completes this direction.

    We next establish that $ \mu_{A \rightarrow B}(G) \leq \log \lambda_{A \rightarrow B}(G).$ This direction is comparatively rather involved whose main ingredient is the following lemma:
    \begin{lemma}
        \label{lem:ent-uniform}
        There exists a $P_X^*$, an optimizer of the outer minimization problem in
        $$ \min_{P_X} \max_{W \in \mathcal{W}(G)} (H(Y) - H(X)), $$
        such that the inner optimizer $W^*(G; P_X^*)$ induces a uniform output distribution on $\Nc(S^*)$.
    \end{lemma}

    Now, let $S^*$ be the support of $P_{X}^*$. If so, one would have
    \begin{align*}
         & \lambda_{A \rightarrow B}(G) = H(Y) - H(X) = \log |\Nc(S^*)| - H(X)  \geq \log \frac{|\Nc(S^*)|}{|S^*|} \geq \min_{S \subseteq A, S \neq \emptyset} \frac{|\mathcal{N}(S)|}{|S|} =  \mu_{A \rightarrow B}(G),
    \end{align*}
    and the proof is complete.
\end{proof}

We will now develop some preliminaries needed to establish \Cref{lem:ent-uniform}.

\begin{definition}
    \label{defn:active}
    Given an input distribution $P_X$ and a bipartite graph $G$, we define an edge $(a, b) \in G$  to be \textit{active} under $W^*(G; P_X)$ if $W^*(b | a) > 0$. Otherwise, it is said to be \textit{inactive}.
\end{definition}

\begin{lemma}
    \label{lem:max-properties}
    Let $S$ be the support of $P_X$.
    \begin{enumerate}
        \item Any maximizer $W^*(G; P_X)$ induces an output distribution, $P_Y$, such that the support of $P_Y$ is $\Nc(S)$.
        \item Let $a_1 \in S$ and $(a_1,b_1), (a_1,b_2)$ be edges in $G$.
              \begin{enumerate}
                  \item If the edges $(a_1, b_1)$ and $(a_1, b_2)$ are active under $W^*(G; P_X)$, then $P_Y(b_1) = P_Y(b_2)$.
                  \item If $(a_1, b_1)$ is active and $(a_1, b_2)$ is inactive under $W^*(G; P_X)$, then $P_Y(b_1) \geq P_Y(b_2)$.
              \end{enumerate}
    \end{enumerate}
\end{lemma}
\begin{proof}
    The proof of part $1)$ proceeds by contradiction. Assume that there exists $b_1 \in \Nc(S)$ such that $P_Y(b_1) = 0$. This implies that these exists  $a_1 \in S$, such that $(a_1, b_1) \in G$ and $W_{Y | X}^*(b_1 | a_1) = 0$ as $P_Y(b_1) = 0$. Further since $P_X(a_1) > 0$, there exists $b_2 \in \Nc(S)$ with $(a_1, b_2) \in G$ and $W_{Y | X}^*(b_2 | a_1) > 0$. For $\alpha \geq 0$ and sufficiently small, define $W_\alpha$ as follows:
    \begin{align*}
        W_{Y | X, \alpha}(b | a) =
        \begin{cases}
            W_{Y | X}^*(b | a) + \alpha = \alpha, & (a, b) = (a_1, b_1) \\
            W_{Y | X}^*(b | a) - \alpha,          & (a, b) = (a_1, b_2) \\
            W_{Y | X}^*(b | a),                   & \text{otherwise}
        \end{cases}.
    \end{align*}
    Define $f(\alpha) := H(Y_\alpha) - H(X)$, where $P_{Y_\alpha}$ is the output distribution of $P_X$ under $W_\alpha$. Note that
    \begin{align*}
        f'(\alpha) = P_X(a_1) \log \left(\frac{P_Y(b_2) - \alpha P_X(a_1)}{\alpha P_X(a_1)}\right).
    \end{align*}
    By assumption, $W_0 = W^*$ is a maximizer of $f(\alpha)$. However, $f'(\alpha) \rightarrow +\infty$ as $\alpha \rightarrow 0^+$, yielding the requisite contradiction.

    We now establish part $2)$.
    Note that $H(Y)$ is concave in $\mathcal{W}(G)$ and all constraints in $\mathcal{W}(G)$ is linear under $\mathcal{W}$. Therefore, Karush–Kuhn–Tucker(KKT) conditions are the necessary and sufficient conditions for optimality for $W_{Y|X}$. We rewrite the optimization problem as follows,
    \begin{align*}
        \begin{array}{cl}
            \min\limits_{W \in \mathcal{W}(G)} & (H(Y) - H(X))                         \\
            \text{subject to}                  & W(b | a) \geq 0, a\in S, (a, b) \in G \\
                                               & \sum_b W(b | a) = 1, a \in S
        \end{array}.
    \end{align*}
    Define the Lagrangian  as follows,
    \begin{align*}
        \mathcal{L}(W) := H(Y) - \mu_{a, b} W(a | b) + \sum_a \lambda_a \left(\sum_b W(b | a) - 1\right).
    \end{align*}
    The KKT conditions for optimality implies that for $W \in \mathcal{W}$, $a \in S$, and $(a,b) \in G$, we have
    \begin{align*}
        \frac{\partial \mathcal{L}}{\partial W(b | a)} = - P_X(a) (\log P_Y(b) + 1) - \mu_{a, b} + \lambda_a & = 0,    \\
        \mu_{a, b} W(b | a)                                                                                  & = 0,    \\
        \mu_{a, b}                                                                                           & \geq 0.
    \end{align*}
    By solving the above conditions, we have
    \begin{align*}
        P_Y(b) & = \exp\left(-\frac{(\tilde{\lambda}_a + \mu_{a, b})}{P_X(a)}\right),
    \end{align*}
    where $\tilde{\lambda}_a =  P_X(a) - \lambda_a$.
    \begin{enumerate}[$a)$]
        \item Suppose $(a_1, b_1)$ and $(a_1, b_2)$ are active. This implies that $\mu_{a_1, b_1} = \mu_{a_1, b_2} = 0$, and forces $P_Y(b_1) = P_Y(b_2)$.
        \item Suppose $(a_1, b_1)$ is active and $(a_1, b_2)$ is inactive. We have $\mu_{a_1, b_1} = 0$ and $\mu_{a_1, b_2} \geq 0$, this implies $P_Y(b_1) \geq P_Y(b_2)$.
    \end{enumerate}
    This establishes part $2)$ of the lemma.
\end{proof}

Based on $P_X$ (with support $S$) and the properties of the maximizer $W^*(G; P_X)$, we induce equivalence relationships between elements in $\Nc(S)$ and between elements in $S$. Let $P_Y$ be the distribution on $\Nc(S)$ induced by $P_X$ and $W^*(G; P_X)$. For $b_1, b_2 \in \Nc(S)$, we say that $b_1 \sim b_2$ if $P_Y(b_1) = P_Y(b_2)$. We use the above to induce an equivalence relationship on $S$ as follows: For $a_1,a_2\in S$, we say that $a_1 \sim a_2$ if there exists $b_1,b_2 \in \Nc(S)$ such that the edges $(a_1,b_1)$ and $(a_2,b_2)$ are active (see \Cref{defn:active}) and $b_1 \sim b_2$.

\begin{remark}
    \label{rem:equivalence-class}
    The main observation is that the active edges in $W^*(G; P_X)$ partition the graph into disconnected components and further there is a one-to-one correspondence between the equivalences classes in $\Nc(S)$ and the equivalence classes in $S$. To see this: consider an equivalence class $T \subset  \Nc(S)$ and let $\hat{S} = \{a \in S: (a,b) ~ \mbox{is active for some} ~ b \in T\}.$ From Lemma \ref{lem:max-properties}, we see that all elements in $\hat{S}$ are equivalent to each other and there is no active edge $(a,b)$ where $a \in \hat{S}$ and $b \notin T$. Further if $a_1 \in S \setminus \hat{S}$, then observe that $a_1$ is not equivalent to any element in $\hat{S}$.
\end{remark}

Let $T_1, \dots ,T_k$ be the partition of $\Nc(S)$ into equivalence classes and let $S_1,\dots ,S_k$ be the corresponding partition of $S$ into equivalence classes.
We can define a total order on the equivalence classes of $\Nc(S)$ as follows: we say $T_{i_1} \geq T_{i_2}$ if $P_Y(b_{i_1}) \geq P_Y(b_{i_2})$. This also induces a total order on the equivalence classes on $S$. Further, without loss of generality, let us assume that $T_1,\dots ,T_k$ (and correspondingly $S_1,\dots ,S_k$) be monotonically decreasing according to the order defined above.

\subsection{Proof of \Cref{lem:ent-uniform}}
\begin{proof}
    Let $P_X^*$ be an optimizer of the outer minimization problem in \eqref{eq:ent-mag} and let $S^*$ be its support. Further, let $S_1, \dots, S_k$ be the equivalence classes (that form a partition of $S$) induced by $W^*(G; P_X^*)$. If $k=1$, i.e. there is only one equivalence class, then Lemma \ref{lem:max-properties} implies that $P_X^*$ and $W^*(G; P_X^*)$ induces a uniform output distribution on $\Nc(S^*)$. Therefore, our goal is to show the existence of an optimizer $P_X^*$ that induces exactly one equivalence class.

    Let $S_1$ and $S_2$ be the largest and second largest elements under the total ordering mentioned previously. Let $m_\ell = |S_\ell|$, $n_\ell = |T_\ell|$, and for $1 \leq i \leq k$, let $s_{i,j}, 1 \leq j \leq m_i$ be an enumeration of the elements of $S_i$ and $t_{i,j}, 1 \leq j \leq n_i$ be an enumeration of the elements of $T_i$. Further let $p_{i,j} = P_{X}^*(s_{i,j})$ and $p_i = \sum_{j=1}^{m_i} p_{i,j}$. Since the induced output probabilities on the elements of $T_i$ are uniform (by the definition of equivalence class), observe that $q_{i,j} := P_Y^*(t_{i,j}) = \frac{p_i}{n_i}$ for all $1 \leq j \leq n_i$.

    By the grouping property of entropy, we have
    \begin{align*}
        H(X) & =H(p_{1,1},..p_{1,m_1},p_{2,1},..,p_{2,m_2},p_{3,1}...,p_{k,m_k})                                                                           \\
             & = p_1 H\left(\frac{p_{1,1}}{p_1}, \dots, \frac{p_{1,m_1}}{p_1}\right) + p_2 H\left(\frac{p_{2,1}}{p_2}, \dots, \frac{p_{2,m_2}}{p_2}\right) \\
             & \quad + (p_1 + p_2) H\left(\frac{p_1}{p_1 + p_2}, \frac{p_2}{p_1 + p_2}\right)                                                              \\
             & \quad + H(p_1 + p_2, p_{3,1}, \dots ,p_{k,m_k}).
    \end{align*}
    Similarly,
    \begin{align*}
        H(Y) & = p_1 H\left(\frac{1}{n_1}, \dots, \frac{1}{n_1}\right) + p_2 H\left(\frac{1}{n_2}, \dots, \frac{1}{n_2}\right) \\
             & \quad + (p_1 + p_2) H\left(\frac{p_1}{p_1 + p_2}, \frac{p_2}{p_1 + p_2}\right)                                  \\
             & \quad + H(p_1 + p_2, q_{3, 1}, \dots, q_{k, n_k}).
    \end{align*}
    Define a parameterized family of input distributions $\tilde{P}_{X(\alpha)}$ as follows:
    \begin{align*}
        \tilde{P}_{X(\alpha)}(s_{i,j}) =
        \begin{cases}
            \left(1 - \frac{\alpha}{p_1}\right) p_{i,j},  & i=1               \\
            \left(1 + \frac{\alpha}{p_2}\right)  p_{i,j}, & i=2               \\
            p_{i,j},                                      & \text{otherwise}.
        \end{cases}
    \end{align*}
    By Lemma \ref{lem:ent-improvement} we know that for $\alpha \in [\alpha_{\min},\alpha_{\max}]$, where
    $$ \alpha_{\max} := \frac{p_1n_2-p_2n_1}{n_1 + n_2} \geq 0 \geq n_2\left(\frac{p_3}{n_3} - \frac{p_{2}}{n_2}\right) =: \alpha_{\min}, $$
    $W^*(G; P_X^*)$ remain the optimal channel. Observe that the induced output distribution is
    \begin{align*}
        \tilde{P}_{Y(\alpha)}(t_{i,j}) =
        \begin{cases}
            \left(1 - \frac{\alpha}{p_1}\right) q_{i,j} = \frac{p_i}{n_i} - \frac{\alpha}{n_i}, & i=1               \\
            \left(1 + \frac{\alpha}{p_2}\right) q_{i,j} = \frac{p_i}{n_i} + \frac{\alpha}{n_i}, & i=2               \\
            q_{i,j},                                                                            & \text{otherwise}.
        \end{cases}
    \end{align*}
    This implies $\lambda_{A \rightarrow B}(G; \tilde{P}_{X(\alpha)}) = H(\tilde{Y}(\alpha)) - H(\tilde{X}(\alpha))$.
    Note that
    \begin{align*}
         & \lambda_{A \rightarrow B}(G; \tilde{P}_{X(\alpha)}) := H(\tilde{Y}(\alpha)) - H(\tilde{X}(\alpha))                                                   \\
         & = (p_1 - \alpha) \left(H\left(\frac{1}{n_1}, \dots, \frac{1}{n_1}\right) - H\left(\frac{p_{1, 1}}{p_1}, \dots, \frac{p_{1, m_1}}{p_1}\right) \right)
        + (p_2 + \alpha) \left(H\left(\frac{1}{n_2}, \dots, \frac{1}{n_2}\right) - H\left(\frac{p_{2, 1}}{p_2}, \dots, \frac{p_{2, m_2}}{p_2}\right)  \right)   \\
         & \quad + H(p_1 + p_2, q_{3, 1}, \dots, q_{k, n_k}) - H(p_1 + p_2, p_{3, 1}, \dots, p_{k, m_k})                                                        \\
         & =(p_1 - \alpha) f_1 + (p_2 + \alpha) f_2 + H(p_1 + p_2, q_{3, 1}, \dots, q_{k, n_k})  - H(p_1 + p_2, p_{3, 1}, \dots, p_{k, m_k}),
    \end{align*}
    where
    \begin{align*}
        f_1 & = H\left(\frac{1}{n_1}, \dots, \frac{1}{n_1}\right) - H\left(\frac{p_{1, 1}}{p_1}, \dots, \frac{p_{1, m_1}}{p_1}\right), \\
        f_2 & = H\left(\frac{1}{n_2}, \dots, \frac{1}{n_2}\right) - H\left(\frac{p_{2, 1}}{p_2}, \dots, \frac{p_{2, m_2}}{p_2}\right).
    \end{align*}
    Thus, $\lambda_{A \rightarrow B}(G; \tilde{P}_{X(\alpha)})$ is linear in $\alpha$.
    At $\alpha=0$, note that $\tilde{P}_{X(\alpha)} = P_X^*$, and hence is a minimizer of $\lambda_{A \rightarrow B}(G; \tilde{P}_{X(\alpha)})$. Therefore, this necessitates that $f_1=f_2$, and for $\alpha \in [\alpha_{\min},\alpha_{max}]$ we have that $\lambda_{A \rightarrow B}(G; \tilde{P}_{X(\alpha)})$ is a constant. Consequently, both $\tilde{P}_{X(\alpha_{\min})}$ and $\tilde{P}_{X(\alpha_{\max})}$ are minimizers of the outer minimization problem.

    If we consider $\tilde{P}_{X(\alpha_{\max})}$ observe that we have $\tilde{P}_{Y(\alpha_{\max})}(t_{1, j}) = \tilde{P}_{Y(\alpha_{\max})}(t_{2, j})$. Therefore $t_{1, j} \sim t_{2, j}$ and this causes $T_1$ and $T_2$ to merge into a new equivalence class. Therefore, we have a minimizer of the outer minimization problem with $k-1$ equivalence classes. We can proceed by induction till we get a single equivalence class. Note that the output elements in an equivalent class have the same probability, and the support of the induced output distribution is the neighborhood of the support of $P_X^*$ (see Lemma \ref{lem:max-properties}). Therefore, establishing that $p_X^*$ induces a single equivalence class establishes Lemma \ref{lem:ent-uniform}.

    \textit{Alternately}, if we consider $\tilde{P}_{X(\alpha_{\min})}$ observe that we have $\tilde{P}_{Y(\alpha_{\max})}(t_{2, j}) = \tilde{P}_{Y(\alpha_{\max})}(t_{3, j})$. Therefore $t_{2, j} \sim t_{3, j}$ and this causes $T_2$ and $T_3$ to merge into a new equivalence class. Therefore, again, we have a minimizer of the outer minimization problem with $k-1$ equivalence classes. Proceeding as above, we can reduce to a single equivalence class and establish Lemma \ref{lem:ent-uniform}.
\end{proof}

\begin{remark}
    The argument above can be used to infer (with minimal modifications) that  any minimizer $P_X^*$  of the outer minimization problem must have $f_i = f_j$, where
    \begin{align*}
        f_i & = H\left(\frac{1}{n_i}, \dots, \frac{1}{n_i}\right) - H\left(\frac{p_{i, 1}}{p_i}, \dots, \frac{p_{i, m_i}}{p_i}\right), \\
        f_j & = H\left(\frac{1}{n_j}, \dots, \frac{1}{n_j}\right) - H\left(\frac{p_{j, 1}}{p_j}, \dots, \frac{p_{j, m_j}}{p_j}\right).
    \end{align*}
    Further $\lambda_{A \rightarrow B}(G; \tilde{P}_{X^*}) = \sum_{i=1}^k p_i f_i $. Since all $f_i$'s are identical, we have $\lambda_{A \rightarrow B}(G) = f_1$. Therefore, the restriction of $\tilde{P}_{X^*}$ to the first equivalence class is also a minimizer of the outer minimization problem, and observe that the induced output is uniform in $T_1$.
\end{remark}

\begin{lemma}[Reweighting input equivalence class probabilities preserves the optimality of the channel]
    \label{lem:ent-improvement}
    Let the partition $S_1 \geq S_2 \geq \dots \ge S_k$ (of $S$, the support of $P_X$) be the monotonically decreasing order of equivalence classes induced by $W^*(G; P_X)$. Define a parameterized family of input distributions $\tilde{P}_{X(\alpha)}$ as follows
    \begin{align*}
        \tilde{P}_{X(\alpha)}(s_{i,j}) =
        \begin{cases}
            \left(1 - \frac{\alpha}{p_1}\right) p_{i,j},  & i=1               \\
            \left(1 + \frac{\alpha}{p_2}\right)  p_{i,j}, & i=2               \\
            p_{i,j},                                      & \text{otherwise}.
        \end{cases}
    \end{align*}
    Then $W^*(G; P_X)$ continues to be an optimal channel under $\tilde{P}_{X(\alpha)}$ for $\alpha \in [\alpha_{\min},\alpha_{\max}]$, where
    $$\alpha_{\max} := \frac{p_1n_2-p_2n_1}{n_1 + n_2} \geq 0 \geq n_2\left(\frac{p_3}{n_3} - \frac{p_{2}}{n_2}\right) =: \alpha_{\min}. $$
\end{lemma}
\begin{proof}
    We recall the KKT conditions (from the proof of Lemma \ref{lem:max-properties}), which are necessary and sufficient for the inner optimization problem to verify the optimality of $W^*(G,P_X)$.
    The KKT conditions for optimality implies that for $a \in S$ and $(a,b) \in G$, we have
    \begin{align*}
        - P_X(a) (\log P_Y(b) + 1) - \mu_{a, b} + \lambda_a & = 0,    \\
        \mu_{a, b} W(b | a)                                 & = 0,    \\
        \mu_{a, b}                                          & \geq 0.
    \end{align*}
    For  $a \in S$ and $(a,b) \in G$, let $\la_a, \mu_{a,b}$ denote the dual parameters that certify the optimality of
    $W^*(G,P_X)$ for $P_X^*$. Now define
    \begin{align*}
        \la_a(\alpha) =
        \begin{cases}
            \left(1 - \frac{\alpha}{p_1}\right) \left(\lambda_a + P_X^*(a) \log\left(1 - \frac{\alpha}{p_1}\right)\right) & a \in S_1         \\
            \left(1 + \frac{\alpha}{p_2}\right) \left(\lambda_a + P_X^*(a) \log\left(1 + \frac{\alpha}{p_2}\right)\right) & a \in S_2         \\
            \lambda_a ,                                                                                                   & \text{otherwise}.
        \end{cases}
    \end{align*}
    Using the channel $W^*(G; P_X)$, the induced output distribution of $\tilde{P}_{X(\alpha)}$, is given by
    \begin{align*}
        \tilde{P}_{Y(\alpha)}(t_{i,j}) =
        \begin{cases}
            \left(1 - \frac{\alpha}{p_1}\right) q_{i,j} = \frac{p_i}{n_i} - \frac{\alpha}{n_i},  & i=1               \\
            \left(1 + \frac{\alpha}{p_2}\right)  q_{i,j} = \frac{p_i}{n_i} + \frac{\alpha}{n_i}, & i=2               \\
            q_{i,j} = \frac{p_i}{n_i},                                                           & \text{otherwise}.
        \end{cases}
    \end{align*}
    Observe that if $(a,b_a)$ is an active edge under $W^*(G; P_X)$, then note that $P_{Y(\alpha)}(b_a)$ only depends on $a$, or rather only on the equivalence class that $a$ (or equivalently $b_a$) belongs to.
    Define
    $$ \mu_{a,b}(\alpha) = P_{X(\alpha)}(a) (\log P_{Y(\alpha)}(b_a) - \log P_{Y(\alpha)}(b)).$$
    Note that $\mu_{a,b}(\alpha) \geq 0$ as long as
    $$1 \geq \frac{p_1}{n_1} - \frac{\alpha}{n_1} \geq \frac{p_2}{n_2} + \frac{\alpha}{n_2}  \geq \frac{p_3}{n_3},$$
    or the ordering of equivalence classes remains unchanged.
    (Note that if $k=2$, i.e. there are only two partitions, then we set $p_3=0$.)
    This is equivalent to $\alpha \geq \max\{n_2\left(\frac{p_3}{n_3} - \frac{p_{2}}{n_2}\right), p_1 - n_1 \}$ and $\alpha \leq \frac{p_1n_2-p_2n_1}{n_1 + n_2}$. Since $n_1 \geq 1$, and by our ordering of equivalence classes, we have $\frac{p_1}{n_1} \geq \frac{p_2}{n_2} \geq \frac{p_3}{n_3}$; a moments reflection implies the following:
    $$ \frac{p_1n_2-p_2n_1}{n_1 + n_2} \geq 0 \geq n_2\left(\frac{p_3}{n_3} - \frac{p_{2}}{n_2}\right) \geq p_1 - n_1. $$
    Therefore $\alpha\in [\alpha_{\min},\alpha_{\max}]$ preserves the ordering of equivalence classes. A simple substitution shows that the dual varaiables $\la_a(\alpha)$ and $\mu_{a,b}(\alpha)$ defined above serve as witnesses for the optimality of $W^*(G; P_X)$ for $P_{X(\alpha)}$. This completes the proof of the lemma.
\end{proof}
\begin{remark}
    The idea of the above proof is the following. The reweighting of the input classes preserves the uniformity of the output probabilities within each equivalent class and the ordering between the output probabilities between equivalent classes. This happens to be the KKT conditions for the maximality of the channel. The limits are achieved with the output probability in an equivalence class equals the value in its adjacent class. At this point, there are potentially multiple optimizers for the inner problem, and the active and inactive edges could be rearranged as you change $\alpha$ further.
\end{remark}

\section{Conclusion}
In this paper, we develop some information-theoretic tools for proving information inequalities by borrowing from similar combinatorial tools developed in additive combinatorics. In reverse, the tools can also be used to generalize some results (or ambient group structure) in additive combinatorics.

\section*{Bibliographic Remarks}
A short version of this paper was submitted (and subsequently published \cite{lan23}) to the IEEE International Symposium on Information Theory in January 2023, with a fuller version (for review of the arguments) made available on authors' websites. The authors subsequently became aware of significant work done in the intersection of information theory and additive combinatorics in the following arXiv uploads \cite{gmt23,ggmt23}. While the results in this article do not solve any open problems of interest, it is hoped that the interplay between these two areas can explored further for the mutual benefit of these two communities. The authors would like to thank the constructive discussions with Lampros Gavalakis and Ioannis Kontoyannis and the anonymous reviewers of our ISIT submission for various comments. The authors would like to thank Profs. Amin Gohari and Cheuk-Ting Li for various helpful discussions. The authors would also like to thank Prof. Ben Green who brought to our attention the manuscript \cite{tav06ds}, which contained an entropic proof of Theorem \ref{th:katz-tao-ent}. (At the time of the original version of this article, the authors were not aware of \cite{tav06ds}, and hence claimed credit for developing an entropic proof of Theorem \ref{th:katz-tao-ent} from its combinatorial counterpart.)

\bibliographystyle{hieeetr}
\bibliography{mybiblio}

\begin{appendices}
\section{Proof of Theorem \ref{th:katz-tao-ent}}
\begin{proof}
    The arguments here are directly motivated by those for establishing the sumset inequality in \cite{kat99} and are essentially identical to the one employed in \cite{tav06ds}. We still present it here to highlight the role played by Lemma \ref{lem:copy-lemma}.
    Consider a joint distribution $(X, Y, Y^\dagger)$ such that $Y \rightarrow X \rightarrow Y^\dagger$ forms a Markov chain and $(X, Y)$ shares the same marginal as $(X, Y^\dagger)$. From Lemma \ref{lem:copy-lemma} (considering $(X,Y) - X - (X,Y^\dagger)$) we have
    \begin{equation}
        H(X, Y, Y^\dagger)  = H(X,Y) + H(X,Y^\dagger) - H(X) = 2 H(X, Y) - H(X).
        \label{eq:prelimlem1}
    \end{equation}
    Here, the last equality comes from the assumption that $(X, Y) \stackrel{(d)}{=} (X, Y^\dagger)$.

    Define three functions: $f_1(x, y, y^\dagger) = (x + y, x + y^\dagger), f_2(x, y, y^\dagger) = (y, y^\dagger), f_3(x, y, y^\dagger) = (x + y, y^\dagger)$.
    Consider a joint distribution of $(X_1, Y_1, Y_1^\dagger, X_2, Y_2, Y_2^\dagger, X_3, Y_3, Y_3^\dagger, X_4, Y_4, Y_4^\dagger)$ such that the following three conditions are satisfied:
    \begin{enumerate}
        \item $(X_i, Y_i, Y_i^\dagger)$ shares the same marginal as $(X, Y, Y^\dagger)$ for $1 \leq i \leq 4$.
        \item $f_i(X_i, Y_i, Y_i^\dagger) = f_{i}(X_{i + 1}, Y_{i + 1}, Y_{i + 1}^\dagger)$ for $1 \leq i \leq 3$.
        \item $(X_1, Y_1, Y_1^\dagger) \rightarrow f_1(X_1, Y_1, Y_1^\dagger) \rightarrow (X_2, Y_2, Y_2^\dagger) \rightarrow f_2(X_2, Y_2, Y_2^\dagger) \rightarrow (X_3, Y_3, Y_3^\dagger) \rightarrow f_3(X_3, Y_3, Y_3^\dagger) \rightarrow (X_4, Y_4, Y_4^\dagger)$ forms a Markov chain.
    \end{enumerate}
    Now by Lemma \ref{lem:copy-lemma}, we have
    \begin{align}
        \label{eqn:kat-first}
        H(X_1, Y_1, Y_1^\dagger, X_2, Y_2, Y_2^\dagger, X_3, Y_3, Y_3^\dagger, X_4, Y_4, Y_4^\dagger) = 4 H(X, Y, Y^\dagger) - H(X + Y, X + Y^\dagger) - H(Y, Y^\dagger) - H(X + Y, Y^\dagger).
    \end{align}
    From condition 2) above and the definition of $f_1,f_2,f_3$, we have the following equalities:
    \begin{align*}
        X_1 + Y_1 = X_2 + Y_2, \quad X_1 + Y_1^\dagger = X_2 + Y_2^\dagger, \quad Y_2 = Y_3, \quad Y_2^\dagger = Y_3^\dagger, \quad  X_3 + Y_3 = X_4 + Y_4, \quad Y_3^\dagger = Y_4^\dagger.
    \end{align*}
    From this, we obtain the following:
    \begin{align*}
        Y_1 - Y_1^\dagger = Y_2 - Y_2^\dagger = Y_3 - Y_3^\dagger.
    \end{align*}
    Consequently, we have
    \begin{align*}
        X_4 - Y_4^\dagger = (X_4 + Y_4) - Y_4 - Y_4^\dagger = (X_3 + Y_3) - Y_4^\dagger - Y_4 = X_3 + (Y_3 - Y_3^\dagger) - Y_4 = X_3 + Y_1 - Y_1^\dagger - Y_4.
    \end{align*}
    Therefore $X_4 - Y_4^\dagger$ is a function of $(X_1, Y_1, Y_1^\dagger, X_3, Y_4)$.
    Therefore,
    \begin{align*}
         & H(X_1, Y_1, Y_1^\dagger, X_2, Y_2, Y_2^\dagger, X_3, Y_3, Y_3^\dagger, X_4, Y_4, Y_4^\dagger | X_1, Y_1, Y_1^\dagger, X_3, Y_4)                                                           \\
         & \quad= H(X_1, Y_1, Y_1^\dagger, X_2, Y_2, Y_2^\dagger, X_3, Y_3, Y_3^\dagger, X_4, Y_4, Y_4^\dagger | X_1, Y_1, Y_1^\dagger, X_3, Y_4, X_4 - Y_4^\dagger)                                 \\
         & \quad = H(X_4 | X_1, Y_1, Y_1^\dagger, X_3, Y_4, X_4 - Y_4^\dagger)                                                                                                                       \\
         & \qquad+ H(X_1, Y_1, Y_1^\dagger, X_2, Y_2, Y_2^\dagger, X_3, Y_3, Y_3^\dagger, X_4, Y_4, Y_4^\dagger | X_1, Y_1, Y_1^\dagger, X_3, X_4, Y_4, Y_4^\dagger)                                 \\
         & \quad \leq H(X_4 | X_4 - Y_4^\dagger) + H(X_1, Y_1, Y_1^\dagger, X_2, Y_2, Y_2^\dagger, X_3, Y_3, Y_3^\dagger, X_4, Y_4, Y_4^\dagger | X_1, Y_1, Y_1^\dagger, X_3, X_4, Y_4, Y_4^\dagger) \\
         & \quad  = H(X_4 | X_4 - Y_4^\dagger) + H( X_2, Y_2, Y_2^\dagger, Y_3, Y_3^\dagger | X_1, Y_1, Y_1^\dagger, X_3, X_4, Y_4, Y_4^\dagger).
    \end{align*}

    To complete the argument, observe that $Y_2 = Y_3 = X_4 + Y_4 - X_3$, $Y_2^\dagger = Y_3^\dagger = Y_4^\dagger$, and $X_2 = X_1 + Y_1 - Y_2 = X_1 + Y_1 + X_3 - X_4 - Y_4$. This implies that $(X_2, Y_2, Y_2^\dagger, Y_3, Y_3^\dagger)$ is a function of $(X_1, Y_1, Y_1^\dagger, X_3, X_4, Y_4, Y_4^\dagger)$. Therefore
    \[ H( X_2, Y_2, Y_2^\dagger, Y_3, Y_3^\dagger | X_1, Y_1, Y_1^\dagger, X_3, X_4, Y_4, Y_4^\dagger) = 0, \]
    implying that
    \begin{align*}
         & H(X_1, Y_1, Y_1^\dagger, X_2, Y_2, Y_2^\dagger, X_3, Y_3, Y_3^\dagger, X_4, Y_4, Y_4^\dagger | X_1, Y_1, Y_1^\dagger, X_3, Y_4) \leq H(X_4 | X_4 - Y_4^\dagger).
    \end{align*}
    Thus, we have
    \begin{align}
        \label{eqn:kat-second}
        H(X_1, Y_1, Y_1^\dagger, X_2, Y_2, Y_2^\dagger, X_3, Y_3, Y_3^\dagger, X_4, Y_4, Y_4^\dagger) \leq H(X_1, Y_1, Y_1^\dagger, X_3, Y_4) + H(X_4 | X_4 - Y_4^\dagger).
    \end{align}

    By using (\ref{eqn:kat-first}) and (\ref{eqn:kat-second}), we have
    \begin{align*}
        0 & \geq 4 H(X, Y, Y^\dagger) - H(X + Y, X + Y^\dagger) - H(Y, Y^\dagger) - H(X + Y, Y^\dagger) - H(X_1, Y_1, Y_1^\dagger, X_3, Y_4) - H(X_4 | X_4 - Y_4^\dagger) \\
          & = 3 H(X, Y, Y^\dagger) - H(X + Y, X + Y^\dagger) - H(Y, Y^\dagger) - H(X + Y, Y^\dagger) - H(X_3, Y_4 | X_1, Y_1, Y_1^\dagger) - H(X_4 | X_4 - Y_4^\dagger).
    \end{align*}
    Now using \eqref{eq:prelimlem1} to replace $H(X,Y,Y^\dagger)$ we have
    \begin{align}
        0 & \geq 6 H(X, Y) - 3 H(X) - H(X + Y, X + Y^\dagger) - H(Y, Y^\dagger) - H(X + Y, Y^\dagger)  - H(X_3, Y_4 | X_1, Y_1, Y_1^\dagger) - H(X_4 | X_4 - Y_4^\dagger)\nonumber \\
          & \geq 6 H(X, Y) - 3 H(X) - 3 H(Y) - 3 H(X + Y) - H(X_3) - H(Y_4) - H(X_4, Y_4^\dagger) + H(X_4 - Y_4^\dagger)\nonumber                                                  \\
          & = 5 H(X, Y) - 4 H(X) - 4 H(Y) - 3 H(X + Y) + H(X - Y)\label{eq:equivform}                                                                                              \\
          & = \frac{1}{2} I(X; X - Y) + \frac{1}{2} I(Y; X - Y) - \frac{3}{2} I(X; X + Y) - \frac{3}{2} I(Y; X + Y) - 3 I(X; Y). \nonumber
    \end{align}
    This completes the proof of the theorem.
\end{proof}

\def\a{\alpha}
\def\b{\beta}
\def\g{\gamma}
\def\d{\delta}
\def\e{\epsilon}
\def\eps{\epsilon}
\def\la{\lambda}
\def\td{\tilde}
\def\o{\omega}
\def\s{\sigma}
\def\ce{\mathfrak{C}}

\def\ty{\mathsf{T}}

\section{Discrete Sanov Theorem}
\label{sec:discrete-sanov}
For the sake of completeness, we will give a detailed proof of discrete Sanov theorem. Let $\Sigma$ be a finite space. For the purposes of this article, let $\Sigma=\{1,2,..,M\}$, or equivalently $\Sigma=[1:M]$. Denote $\Mc(\Sigma)$ as the set of probability mass functions on $\Sigma$. For a given probability mass function $\mu$, let us denote $\Sigma_\mu=\{i:\mu(i)>0\}$ to be the support of $\mu$. Thus $\Sigma_\mu \subseteq \Sigma$.

\subsection{Types}
Given a sequence $y^n \in \Sigma^n$, we define the \textbf{type} of $y^n$, $\ty_{y^n} \in \Mc(\Sigma)$, as the probability mass function given by
\[ \ty_{y^n}(i) := \frac{1}{n} \sum_{k=1}^n 1_{\{y_k=i\}}, ~~ 1\leq i \leq M.\]
It is, equivalently, the empirical measure induced by the sequence $y^n$. Let $\ty_n \subset \Mc(\Sigma)$ denote the collection of all types, i.e.
\[ \ty_{n}:=\{\mu: \mu = \ty_{y^n} ~\mbox{for some}~ y^n \in \Sigma^n \}.\]

\begin{lemma}
\label{lem:type-el}
The following statements hold:
\begin{enumerate}[$(i)$]
    \item $|\ty_n|= \binom{n+M-1}{M-1} \leq (n+1)^{M-1}.$ 
    \item For any $\mu \in \Mc(\Sigma)$, there exists $\nu \in \ty_n$ such that $|\mu(i) - \nu(i)|\leq \frac{1}{n}$. Consequently, $d_{TV}(\mu,\ty_n)  \leq \frac{M}{2n},$ where $d_{TV}(\mu,\ty_n) := \min_{\nu \in \ty_n} d_{TV}(\mu,\nu)$ and $d_{TV}(\mu,\nu) = \frac{1}{2} \sum_{i=1}^M |\mu(i)-\nu(i)|$. Further $\Sigma_\nu \subseteq \Sigma_\mu$.
\end{enumerate}
\end{lemma}
\begin{proof}
Note that every $\mu \in \ty_n$  is in one-to-one correspondence with non-negative integer sequences $\{a_1,..,a_M\}$ such that $\sum_{i=1}^M a_i=n$. The count of the latter is a problem in elementary combinatorics, and the count is essentially a bijection to choosing the identities of $M-1$ dividers from $n+M-1$ locations. Note that $\binom{n+M-1}{M-1} \leq (n+1)^{M-1}$ is an equality for $M=1$ and for $M>1$, we have \[ \binom{n+M-1}{M-1} = \prod_{k=1}^{M-1} \frac{n+k}{k} \leq \prod_{k=1}^{M-1} (n+1) = (n+1)^{M-1}, \]
as $\frac{n+k}{k} \leq n+1, \forall k \geq 1$.

Given a $\mu \in \Mc(\Sigma)$, let us define two non-negative integer sequences according to
\[ k_l(i) = \lfloor n \mu(i) \rfloor, \quad k_u(i) = \lceil n \mu(i) \rceil, \quad 1 \leq i \leq M.\]
The following estimates are clear:
\begin{align*}
n \mu(i)  - 1 \leq k_l(i) \leq n \mu(i) \leq k_u(i) \leq n \mu(i)+1.
\end{align*}
Summing up over $i$, we obtain 
\[n-M \leq \sum_{i=1}^M k_l(i) \leq n \leq \sum_{i=1}^M k_u(i) \leq n+M.\]

Therefore, we can find a sequence of non-negative integers, $k_{int}(i)$ such that $k_l(i) \leq k_{int}(i) \leq k_u(i)$ such that $\sum_{i=1}^M k_{int}(i)=n.$ This is essentially like a discrete intermediate-value-theorem, which a greedy algorithm (starting from $k_l$ and increasing value at each coordinate by one while obeying the bounds) can easily establish. Now define $\nu(i) = \frac{k_{int}(i)}{n}$. Note that $\nu\in \ty_n$. We know that
\[n \mu(i)  - 1\leq  k_l(i) \leq k_{int}(i) \leq k_u(i)\leq n \mu(i)+1.\]
Therefore $|\nu(i)-\mu(i)|\leq \frac{1}{n}$ and $d_{TV}(\mu,\nu) \leq \frac{M}{2n}.$ If $\mu(i)=0$, then observe that $k_u(i)=0$ implying $k_{int}(i)=0$ and hence $\nu(i)=0$. This establishes the relationship between the supports.
\end{proof}

For $\nu \in \ty_n$, we define the $\textbf{type-class}$ by
$\Yc_n(\nu) = \{y^n \in \Sigma^n: \ty_{y^n}=\nu\}.$ Note that $\Yc_n(\nu)$ is the collection of permutations of a generic string $y^n$ whose empirical measure is $\nu$, and the cardinality of $\Yc_n(\nu)$ is the multinomial co-efficient associated with the empirical counts, i.e. $|\Yc_n(\nu) |=\binom{n}{n\nu(1), n\nu(2), \cdots, n\nu(M)}$.

Let $\prob_\mu$ be the probability law associated with an infinite sequence of i.i.d. random variables $Y_1,Y_2,..$, distributed according to $\mu \in \mathcal{M}(\Sigma)$.

In the following:
\begin{align*}
    H(\nu) &= \sum_{i=1}^N -\nu(i) \log_2 \nu(i),\\
    D(\nu\|\mu) &= \sum_{i=1}^N \nu(i) \log_2 \frac{\nu(i)}{\mu(i)},
\end{align*}
with the convention: $0 \log_2 0 = 0$ and if $\nu \centernot{\ll} \mu$, then $D(\nu\|\mu)=\infty$.

\begin{lemma}
\label{lem:typprob}
$$\prob_\mu[(Y_1,Y_2,..,Y_n)=y^n] = 2^{-n (H(\ty_{y^n}) + D(\ty_{y^n}\|\mu)) }. $$
\end{lemma}
\begin{proof}
    Note that
    \begin{align*}
    \prob_\mu[(Y_1,Y_2,..,Y_n)=y^n]&= \prod_{i=1}^M (\mu(i))^{n\ty_{y^n}(i)}\\
    & = 2^{-n\left(-\sum_{i=1}^M \ty_{y^n}(i) \log_2 \ty_{y^n}(i) + \sum_{i=1}^M \ty_{y^n}(i) \log_2 \frac{\ty_{y^n}(i)}{\mu(i)} \right)}\\
    & = 2^{-n (H(\ty_{y^n}) + D(\ty_{y^n}\|\mu)) }.
    \end{align*}
\end{proof}

\begin{lemma}
\label{lem:elcomb}
    Let $m,l \in \mathbb{N}$. Then $\frac{m!}{l!} \leq l^{m-l}.$
\end{lemma}
\begin{proof}
    If $m > l$, then $\frac{m!}{l!} = \prod_{k=l+1}^m k \geq l^{m-l}.$ If $m < l$, then $\frac{m!}{l!} = \prod_{k=m+1}^l \frac{1}{k} \geq \frac{1}{l}^{l-m} = l^{m-l}.$ Finally, equality holds for $m=l$.
\end{proof}

\begin{corollary}
    \label{coro:cometype}
    For $\gamma,\nu \in \ty_n$, 
    \[ \frac{|\Yc_n(\nu) |}{|\Yc_n(\gamma) |}\geq  2^{n(H(\nu) - D(\gamma\|\nu) - H(\gamma))}.\]
\end{corollary}
\begin{proof}
Note that
\[ \frac{|\Yc_n(\nu) |}{|\Yc_n(\gamma) |} = \frac{\binom{n}{n\nu(1), n\nu(2), \cdots, n\nu(M)}}{\binom{n}{n\gamma(1), n\gamma(2), \cdots, n\gamma(M)}} = \prod_{i=1}^M \frac{(n\gamma(i))!}{(n\nu(i))!} \geq \prod_{i=1}^M (n\nu(i))^{n(\gamma(i)-\nu(i))} = \prod_{i=1}^M (\nu(i))^{n(\gamma(i)-\nu(i))}. \]

It is immediate  that,
\[ \prod_{i=1}^M (\nu(i))^{n(\gamma(i)-\nu(i))} =  2^{n(H(\nu) - D(\gamma\|\nu) - H(\gamma))}.\]
\end{proof}

\begin{lemma}
    \label{lem:typecount}
For every $\nu \in \ty_n$,
\[ \frac{1}{|\ty_n|} 2^{nH(\nu)} \leq |\Yc_n(\nu)| \leq 2^{nH(\nu)}\]
\end{lemma}

\begin{proof}
    $\prob_\nu$ be the probability law associated with an infinite sequence of i.i.d. random variables $Y_1,Y_2,..$, distributed according to $\nu$.
    \[ \sum_{y^n \in \Yc_n(\nu)}\prob_\nu[(Y_1,Y_2,..,Y_n)=y^n] \leq 1, \]
    implying (from \Cref{lem:typprob}) that
    \[ |\Yc_n(\nu)| 2^{-nH(\nu)} \leq 1. \]

    Now, we also have
    \begin{align*}
    1&= \sum_{\gamma \in \ty_n} \sum_{y^n \in \Yc_n(\gamma)}\prob_\nu[(Y_1,Y_2,..,Y_n)=y^n]\\
    & = \sum_{\gamma \in \ty_n} |\Yc_n(\gamma)| 2^{-n (H(\gamma) + D(\gamma\|\nu)) }\\
    & \stackrel{(a)}{\leq} \sum_{\gamma \in \ty_n} |\Yc_n(\nu)|2^{-n(H(\nu) - D(\gamma\|\nu) - H(\gamma))} 2^{-n (H(\gamma) + D(\gamma\|\nu)) }\\
    & =  \sum_{\gamma \in \ty_n} |\Yc_n(\nu)| 2^{-n H(\nu)  }\\
    &= |\ty_n| |\Yc_n(\nu)| 2^{-n H(\nu)  }.
    \end{align*}
    Here, $(a)$ follows from \Cref{coro:cometype}.
\end{proof}

\begin{lemma} 
\label{le:typegenprob}
For any $\nu, \mu \in \ty_n$,
    \[\frac{1}{|\ty_n|} 2^{-n D(\nu\|\mu)} \leq \prob_{\mu}(\ty_{Y^n}=\nu) \leq 2^{-n D(\nu\|\mu)}. \]
\end{lemma}
\begin{proof}
    From \Cref{lem:typprob}, we see that
    \[ \prob_{\mu}(\ty_{y^n}=\nu) = |\Yc_n(\nu)| 2^{-n(H(\nu) + D(\nu\|\mu) )}. \]
    The proof is completed by applying \Cref{lem:typecount}.
\end{proof}

\begin{lemma}
\label{le:Dcon}
    Let $\Sigma_{\nu_n}, \Sigma_{\nu} \subseteq \Sigma_\mu$ and $\nu_n \to \nu$. Then $D(\nu_n\|\mu) \to D(\nu\|\mu).$
\end{lemma}
\begin{proof}
This follows as, for $i=1,..,M$, $\nu_n(i) \to \nu(i)$ and hence $\nu_n(i) \log \frac{\nu_{n}(i)}{\mu(i)} \to \nu(i)\log \frac{\nu(i)}{\mu(i)}.$
\end{proof}

\begin{theorem}[Sanov] For every $\Gamma \subseteq \Mc(\Sigma_\mu)$,
\[ -\inf_{\nu \in \Gamma^r} D(\nu\|\mu) \leq \liminf_n \frac{1}{n} \log \prob_{\mu}(\ty_{Y^n}\in \Gamma) \leq \limsup_n \frac{1}{n} \log \prob_{\mu}(\ty_{Y^n}\in \Gamma) \leq -\inf_{\nu \in \Gamma} D(\nu\|\mu). \]
Here, $\Gamma^r = \{ \nu \in \Gamma: \exists \nu_n \in \ty_n \cap \Gamma, \nu_n \to \nu\}$.



\end{theorem}
\begin{proof}
    From \Cref{le:typegenprob}, we have
    \begin{align*}
        \prob_{\mu}(\ty_{Y^n}\in \Gamma) &=\sum_{\gamma \in \ty_n \cap \Gamma} \prob_{\mu}(\ty_{Y^n}= \gamma)\\
        &\leq \sum_{\gamma \in \ty_n \cap \Gamma} 2^{-n D(\gamma\|\mu)}\\
        & \leq \sum_{\gamma \in \ty_n \cap \Gamma} 2^{-n \inf_{\nu \in \Gamma} D(\nu\|\mu)}\\
        & = |\Gamma \cap \ty_n|  2^{-n \inf_{\nu \in \Gamma} D(\nu\|\mu)}\\
        & \leq (n+1)^{(M-1)} 2^{-n \inf_{\nu \in \Gamma} D(\nu\|\mu)}.
    \end{align*}
    Therefore, 
    \[ \frac{1}{n} \log \prob_{\mu}(\ty_{Y^n}\in \Gamma) \leq \frac{M-1}{n} \log (n+1) - \inf_{\nu \in \Gamma} D(\nu\|\mu). \]
    Taking $\limsup_{n}$ on both sides yields the upper bound.

    Given $\nu \in \Gamma^r$. Let $\hat{\nu}_n \in  \ty_n \cap \Gamma$, such that $\hat{\nu}_n \to \nu$. Then
    \begin{align*}
        \prob_{\mu}(\ty_{Y^n}\in \Gamma) &=\sum_{\gamma \in \ty_n \cap \gamma} \prob_{\mu}(\ty_{Y^n}= \gamma)\\
        &\geq \prob_{\mu}(\ty_{Y^n}= \hat{\nu}_n)\\
        & \geq \frac{1}{|\ty_n|} 2^{-n D(\hat{\nu}_n\|\mu)},
    \end{align*}
    where the last inequality follows from \Cref{le:typegenprob}.
    
    Therefore
     \[ \frac{1}{n} \log \prob_{\mu}(\ty_{Y^n}\in \Gamma) \geq -\frac{1}{n} \log |\ty_n|  -  D(\hat{\nu}_n\|\mu) \geq -\frac{M-1}{n} \log(n+1) -  D(\hat{\nu}_n\|\mu). \]
     Taking $\liminf$ and using \Cref{le:Dcon}, we obtain
     \[ \liminf_n \frac{1}{n} \log \prob_{\mu}(\ty_{Y^n}\in \Gamma) \geq  -  D(\nu\|\mu).\]
     Since, this holds for all $\nu \in \Gamma^r$, we get
     \[ \liminf_n \frac{1}{n} \log \prob_{\mu}(\ty_{Y^n}\in \Gamma) \geq  - \inf_{\nu\in \Gamma^r} D(\nu\|\mu).\]
\end{proof}

Given two sets $A$ and $B$, the Hausdorff distance is defined as
$$ d_H(A,B) = \max \{\sup_{x\in A} \inf_{y\in B} d(x,y), \sup_{y\in B} \inf_{x\in A} d(x,y)\}.$$
In the space of probability distributions, let us consider the underlying metric to be the total variation distance.

\begin{lemma}
\label{le:limsaDcon}
Let $\{\Gamma_n\}_{n\geq 1}, \Gamma \subseteq \Mc(\Sigma_\mu)$, and $d_H(\Gamma_n,\Gamma)\to 0$. Then
    \[ \lim_n \inf_{\nu \in \Gamma_n} D(\nu\|\mu) = \inf_{\nu \in \Gamma}  D(\nu\|\mu).\]
\end{lemma}
\begin{proof}
     Let $\nu^* \in \Gamma$ be such that $D(\nu^*||\mu) \leq \inf_{\nu \in \Gamma}  D(\nu||\mu) + \epsilon.$ Note that $d_H(\Gamma_n,\Gamma) \geq \inf_{\hat{\nu} \in \Gamma_n}  d_{TV}(\hat{\nu},\nu^*)$. Since $d_H(\Gamma_n,\Gamma) \to 0$, there exists a sequence $\hat{\nu}_n \in \Gamma_n$ such that $\hat{\nu}_n \to \nu^*. $ Hence from Lemma \ref{le:Dcon}, $\lim_n D(\hat{\nu}_n||\mu) \to D(\nu^*||\mu)$. Now 
    \[ \limsup_n \inf_{\nu \in \Gamma_n} D(\nu||\mu) \leq  \limsup_n D(\hat{\nu}_n||\mu) = D(\nu^*||\mu) \leq \inf_{\nu \in \Gamma}  D(\nu||\mu) + \epsilon, \]
    implying
    $ \limsup_n \inf_{\nu \in \Gamma_n} D(\nu||\mu) \leq \inf_{\nu \in \Gamma}  D(\nu||\mu)$.

    Let $n_k$ be a subsequence such that $\inf_{\nu \in \Gamma_{n_k}} D(\nu\|\mu) \stackrel{k\to\infty}{\to} \liminf_n \inf_{\nu \in \Gamma_n} D(\nu\|\mu)$. Consider $\nu_k \in \Gamma_{n_k}$ such that $D(\nu_k\|\mu) \leq \inf_{\nu \in \Gamma_{n_k}} D(\nu,\mu) + \frac{\epsilon}{k}$.
    Since $\mathcal{M}(\Sigma_\mu)$ is compact, there exists a convergent subsequence $\{k_l\}$, i.e. $\nu_{k_l} \to \nu^*$, for some $\nu^* \in \mathcal{M}(\Sigma_\mu)$. By construction,
    \[\limsup_{l}  D(\nu_{k_l}\|\mu) = \limsup_{l}  D(\nu_{k_l}\|\mu) - \frac{\eps}{k_l} \leq\limsup_{k} D(\nu_{k}\|\mu) - \frac{\eps}{k} \leq \limsup_k \inf_{\nu \in \Gamma_{n_k}} D(\nu,\mu) = \liminf_n \inf_{\nu \in \Gamma_n} D(\nu,\mu).  \]
    Since $\nu_{k_l} \to \nu^*$, \Cref{le:Dcon} yields $D(\nu_{k_l}\|\mu) \to D(\nu^*\|\mu).$ Therefore
    \[D(\nu^*\|\mu) \leq \liminf_n \inf_{\nu \in \Gamma_n} D(\nu,\mu). \]

     As $d_H(\Gamma_{n_k},\Gamma) \geq \inf_{\hat{\nu} \in \Gamma}  d_{TV}(\hat{\nu},\nu_k)$, let $\nu^\dagger_k \in \Gamma$, satisfy $d_{TV}(\nu^\dagger_k,\nu_k) \leq d_H(\Gamma_{n_k},\Gamma) + \frac{\epsilon}{k}.$ 
    Note that
    \[d_{TV}(\nu^\dagger_{k_l},\nu^*) \leq d_{TV}(\nu^\dagger_{k_l},\nu_{k_l}) + d_{TV}(\nu_{k_l},\nu^*) \leq d_H(\Gamma_{n_{k_l}},\Gamma) + \frac{\epsilon}{k_l} + d_{TV}(\nu_{k_l},\nu^*).\]
    Therefore, taking $l \to \infty$, we see that $\nu^\dagger_{k_l} \to \nu^*$. Finally as $\nu^\dagger_{k_l} \in \Gamma$,
    \[ \inf_{\nu \in \Gamma} D(\nu\|\mu) \leq \liminf_l D(\nu^\dagger_{k_l}\|\mu) = D(\nu^*\|\mu).\]
    Putting all this together, we obtain
    \[ \inf_{\nu \in \Gamma} D(\nu|\mu) \leq  D(\nu^*\|\mu) \leq  \liminf_n \inf_{\nu \in \Gamma_n} D(\nu,\mu) \leq \limsup_n \inf_{\nu \in \Gamma_n} D(\nu,\mu) \leq \inf_{\nu \in \Gamma} D(\nu,\mu), \]
    establishing the lemma.
\end{proof}

\begin{theorem}[limiting Sanov]
\label{th:limSan}Let $\{\Gamma_n\}_{n\geq 1}, \Gamma \subseteq \Mc(\Sigma_\mu)$, and $d_H(\Gamma_n,\Gamma)\to 0$.
Then,
\begin{equation} -\inf_{\nu \in \Gamma^r} D(\nu\|\mu) \leq \liminf_n \frac{1}{n} \log \prob_{\mu}(\ty_{Y^n}\in \Gamma_n) \leq \limsup_n \frac{1}{n} \log \prob_{\mu}(\ty_{Y^n}\in \Gamma_n) \leq -\inf_{\nu \in \Gamma} D(\nu\|\mu). 
\label{eq:limSan}
\end{equation}
Here, $\Gamma^r = \{ \nu \in \Gamma: \exists \nu_n \in \ty_n \cap \Gamma_n, \nu_n \to \nu\}$. Clearly $\Gamma^r \supseteq \Gamma^o$, where $\Gamma^o$ is the interior of $\Gamma$ considered as a subset of $\Mc(\Sigma_\mu)$. In particular, if $\Gamma_n \subseteq \ty_n$, then $\Gamma^r=\Gamma$, and 
\[  \lim_n \frac{1}{n} \log \prob_{\mu}(\ty_{Y^n}\in \Gamma_n) = -\inf_{\nu \in \Gamma} D(\nu\|\mu). \]
\end{theorem}

\begin{proof}
    Note that
    \begin{align*}
        \prob_{\mu}(\ty_{Y^n}\in \Gamma_n) &=\sum_{\gamma \in \ty_n \cap \Gamma_n} \prob_{\mu}(\ty_{Y^n}= \gamma)\\
        &\stackrel{(a)}{\leq} \sum_{\gamma \in \ty_n \cap \Gamma_n} 2^{-n D(\gamma\|\mu)}\\
        & \leq \sum_{\gamma \in \ty_n \cap \Gamma_n} 2^{-n \inf_{\nu \in \Gamma_n} D(\nu\|\mu)}\\
        & = |\Gamma_n \cap \ty_n|  2^{-n \inf_{\nu \in \Gamma_n} D(\nu_n\|\mu)}\\
        & \leq (n+1)^{(M-1)} 2^{-n \inf_{\nu \in \Gamma_n} D(\nu\|\mu)}.
    \end{align*}
    Here $(a)$ follows from \Cref{le:typegenprob}.
    Therefore, 
    \[ \frac{1}{n} \log \prob_{\mu}(\ty_{Y^n}\in \Gamma_n) \leq \frac{M-1}{n} \log (n+1) - \inf_{\nu \in \Gamma_n} D(\nu\|\mu). \]
    Taking $\limsup_{n}$ on both sides and using \Cref{le:limsaDcon} yields the upper bound.

    Given $\nu \in \Gamma^r$, let $\hat{\nu}_n \in  \ty_n \cap \Gamma, \hat{\nu}_n \to \nu$. Then
    \begin{align*}
        \prob_{\mu}(\ty_{Y^n}\in \Gamma_n) &=\sum_{\gamma \in \ty_n \cap \Gamma_n} \prob_{\mu}(\ty_{Y^n}= \gamma)\\
        &\geq \prob_{\mu}(\ty_{Y^n}= \hat{\nu}_n)\\
        & \stackrel{(a)}{\geq} \frac{1}{|\ty_n|} 2^{-n D(\hat{\nu}_n\|\mu)}.
    \end{align*}
    Here, again, $(a)$ follows from \Cref{le:typegenprob}.
    Therefore
     \[ \frac{1}{n} \log \prob_{\mu}(\ty_{Y^n}\in \Gamma_n) \geq -\frac{1}{n} \log |\ty_n|  -  D(\hat{\nu}_n\|\mu). \]
     Taking $\liminf$ and using \Cref{le:Dcon}, we obtain
     \[ \liminf_n \frac{1}{n} \log \prob_{\mu}(\ty_{Y^n}\in \Gamma_n) \geq  -  D(\nu\|\mu).\]
     Since, this holds for all $\nu \in \Gamma^r$, we get
     \[ \liminf_n \frac{1}{n} \log \prob_{\mu}(\ty_{Y^n}\in \Gamma_n) \geq  - \inf_{\nu\in \Gamma^r} D(\nu\|\mu).\]
     This proves \eqref{eq:limSan}.

     If $\Gamma_n \subseteq \ty_n$, then as $d_H(\Gamma_n,\Gamma) \to 0$ it is immediate that for every $\nu \in \Gamma, \exists \nu_n \in \Gamma_n = \Gamma_n \cap \ty_n$ such that $d_{TV}(\nu_n,\nu)\to 0$. This implies that $\Gamma^r = \Gamma$. 
\end{proof}

\section{Maximal couplings}
\label{sec:max-coup}
In this section, we will show the derivation of maximal coupling from the discrete Sanov theorem. Let $p_X, p_Y$ be two distributions supported on a finite alphabet $\Sigma$. Let $\{\omega_n\}$ be a non-negative sequence such that $\omega_n \to 0$ as $n \to \infty$ and $\omega_n \sqrt{n} \to \infty$ as $n \to \infty$.
Define $A_n \subseteq \Sigma^n$ as
$$A_n = \{ g^n \in \Sigma^n: |\ty_{g^n}(a)-p_X(a)| \leq  \omega_n p_X(a) , \forall a \in \Sigma\}.$$
\begin{remark}
The set $A_n$ is usually called the typical sequences (or strongly-typical sequences) corresponding to distribution $p_X$. These sets play an important role in network information theory, particularly in the proofs of the channel coding theorems.
\end{remark}
Similarly, let
$$B_n = \{ g^n \in \Sigma^n: |\ty_{g^n}(a)-p_Y(a)| \leq \omega_n p_Y(a) ,  \forall a \in \Sigma\}.$$

Define $\hat{\Gamma}_n = \{\nu \in \mathcal{M}(\Sigma\times \Sigma): \nu=\ty_{(g_1^n,g_2^n)}~\mbox{for some}~(g_1^n,g_2^n) \in A_n \times B_n\}$ and 
$\hat{\Gamma} = \Pi(p_X,p_Y)$, the set of all couplings with the given marginals. 

\begin{lemma}
    \label{le:limtypical} Let $\hat{\Gamma}_n$ and $\hat{\Gamma}$ be as described above. Then
    \[ d_H(\hat{\Gamma}_n,\hat{\Gamma}) \to 0.\]
\end{lemma}
\begin{proof}
    Let $\nu \in \hat{\Gamma}$. Then we know, from \Cref{lem:type-el} that there exists $\hat{\nu}_n \in \ty_n$  such that $|\hat{\nu}_n(a,b) - \nu(a,b)| \leq \frac{1}{n}$ for all $(a,b) \in \Sigma_A \times \Sigma_B$, and if $\nu(a,b)=0$, then $\hat{\nu}_n(a,b)=0$.   Now 
    \[ \left|\sum_{b \in \Sigma_B} (\hat{\nu}_n(a,b) - \nu(a,b))\right| \leq \sum_{b \in \Sigma_B} |\hat{\nu}_n(a,b) - \nu(a,b)| \leq \frac{|\Sigma_B|}{n}.  \]
    Note that $\sqrt{n}\omega_n \nu(a,b) \to \infty$, for all $(a,b): \nu(a,b) > 0$. Therefore,
    for large $n$, $\frac{\Sigma_B}{n} \leq \omega_n p_X(a)= \omega_n \sum_{b}\nu(a,b)$.  Similarly, for large $n$,
     \[ \left|\sum_{a \in \Sigma_A} (\hat{\nu}_n(a,b) - \nu(a,b))\right| \leq \sum_{a \in \Sigma_A} |\hat{\nu}_n(a,b) - \nu(a,b)| \leq \frac{|\Sigma_A|}{n} \leq \omega_n p_Y(b)= \omega_n, \sum_{a}\nu(a,b).  \]
     Therefore, for large $n$, any $(g_1^n,g_2^n)$, such that $\ty_{g_1^n,g_2^n} = \hat{\nu}_n(a,b)$, is an element of $A_n \times B_n$, or that $\hat{\nu}_n \in \hat{\Gamma}_n$. Further, note that, $d_{TV}(\hat{\nu}_n,\nu) \leq \frac{|\Sigma_A|| \Sigma_B|}{2n}$. Since this holds for any $\nu \in \hat{\Gamma}$, we obtain that $\sup_{\nu \in \hat{\Gamma}} \inf_{\nu_n \in \hat{\Gamma}_n} d_{TV}(\nu,\nu_n) \to 0$ as $n \to \infty$.

     Now suppose $\sup_{\nu_n \in \hat{\Gamma}_n} \inf_{\nu \in \hat{\Gamma}} d_{TV}(\nu,\nu_n) \centernot{\to} 0$. There, there is a subsequence $n_k$ and $\epsilon >0$ such that $\sup_{\nu_{k} \in \hat{\Gamma}_{n_k}} \inf_{\nu \in \hat{\Gamma}} d_{TV}(\nu,\nu_{k}) > \epsilon$. Therefore, there is a sequence $\nu_k \in \hat{\Gamma}_{n_k}$ such that $\inf_{\nu \in \hat{\Gamma}} d_{TV}(\nu_k,\nu) > \frac{\epsilon}{2}$.  As $\mathcal{M}(\Sigma_A \times \Sigma_B)$ is a compact set, and $\hat{\Gamma}_{n_k} \subset \mathcal{M}(\Sigma_A \times \Sigma_B)$, we have a convergent subsequence $\nu_{k_l} \to \nu^\dagger$. By definition of $\hat{\Gamma}_{n_{k_l}}$  we have
     \[ \left|\left(\sum_b\nu_{k_l}(a,b)\right) - p_X(a) \right| \leq \omega_{n_{k_l}} p_X(a). \]
     As $\omega_{n_{k_l}} \to 0$, $\sum_b \nu^\dagger(a,b) = p_X(a)$. Similarly,  $\sum_a \nu^\dagger(a,b) = p_Y(b)$. Therefore $\nu^\dagger \in \hat{\Gamma}$. Therefore $d_{TV}(\nu_k,\nu^\dagger) \to 0$, contradicting $\inf_{\nu \in \hat{\Gamma}} d_{TV}(\nu_k,\nu) > \frac{\epsilon}{2}$. This shows that $\sup_{\nu_n \in \hat{\Gamma}_n} \inf_{\nu \in \hat{\Gamma}} d_{TV}(\nu,\nu_n) \to 0$ as desired.
\end{proof}

\begin{lemma}[Data-Processing]
    \label{le:dataproc}
    Let $W_{Y|X}$ be a stochastic mapping (channel). Let $p_X,q_X$ be two distributions on $\mathcal{X}$ and $p_Y = \sum_{x} W_{Y|X}p_X$ and $q_Y = \sum_{X} W_{Y|X} q_Y$ be the two induced distributions on $\mathcal{Y}$. Then $d_{TV}(p_X,q_X) \geq d_{TV}(p_Y,q_Y)$.
\end{lemma}
\begin{proof} Observe the following:
\begin{align*}
    \sum_{y} | p_Y(y) - q_Y(y)| & = \sum_{y} \left|\sum_{x} W_{Y|X} (y|x)(p_X(x) - q_X(x)) \right|\\
    & \leq \sum_{x,y} W_{Y|X}(y|x) |p_X(x) - q_X(x)|\\
    & = \sum_x |p_X(x) - q_X(x)|
  \end{align*}
\end{proof}

\begin{theorem}
Let $p_X$ and $p_Y$ be two distributions having finite support on an Abelian group $\mathbb{G}$.  Let $\{\omega_n\}$ be a non-negative sequence such that $\omega_n \to 0$ as $n \to \infty$ and $\omega_n \sqrt{n} \to \infty$ as $n \to \infty$.
Define $A_n \subseteq \mathbb{G}^n$, as
$$A_n = \{ g^n \in \mathbb{G}^n: |k_{g^n}(a)-np_X(a)| \leq n p_X(a) \omega_n, \forall a \in \mathbb{G}\}.$$
Here $k_{g^n}(a):=|\{i:g_i=a, 1 \leq i \leq n\}|.$
Similarly, let
$$B_n = \{ g^n \in \mathbb{G}^n: |k_{g^n}(a)-np_Y(a)| \leq n p_Y(a) \omega_n,  \forall a \in \mathbb{G}\}.$$
    $$\lim_{n} \frac{1}{n} \log |A_n + B_n| = \max_{q\in \Pi(p_X,p_Y)} H_q(X+Y).$$
    Here $\Pi(p_X,p_Y)$ is the set of joint distributions (couplings) such that the marginals are $p_X$ and $p_Y$, respectively.
\end{theorem}
\begin{proof}
     Let $S$ be a finite subset of $\mathbb{G}$ such that $S$ contains support of $p_X, p_Y$ and $p_{X+Y}$. Let $\mu$ be the uniform distribution on $S$. 
Let $C_n=A_n + B_n$, and $\Gamma_n = \{\nu \in \mathcal{M}(S): \nu=\ty_{g^n}~\mbox{for some}~g^n \in C_n\}$. Let
$\Gamma = \{\nu\in \mathcal{M}(\mathbb{G}): \nu = q_{X+Y}, q_{X,Y} \in \Pi(p_X,p_Y)\}.$ Define $\hat{\Gamma}_n = \{\nu \in \mathcal{M}(S\times S): \nu=\ty_{(g_1^n,g_2^n)}~\mbox{for some}~(g_1^n,g_2^n) \in A_n \times B_n\}$ and 
$\hat{\Gamma} = \Pi(p_X,p_Y)$, the set of all couplings with the given marginals.

For $\nu \in \ty_n$, we had defined the type-class by
$\Yc_n(\nu) = \{y^n \in S^n: \ty_{y^n}=\nu\}.$ Hence, $\nu \in \Gamma_n$ if and only if $\mathcal{Y}_n(\nu) \subseteq C_n$. Note that, by definition, the sets $A_n$, $B_n$, and $C_n$ are permutation invariant. Therefore
\[ \prob_{\mu}(\ty_{Y^n}\in \Gamma_n) = \prob_{\mu}({Y^n}\in C_n) = \frac{|C_n|}{|S|^n}. \]
From Lemma \ref{le:limtypical}, $d_H(\hat{\Gamma}_n,\hat{\Gamma}) \to 0$, and since $\hat{\Gamma}_n \subseteq \ty_n$, $\hat{\Gamma}=\hat{\Gamma}^r$. Considering $(X,Y) \mapsto X+Y$, by Lemma \ref{le:dataproc}, we obtain $d_H(\Gamma_n,\Gamma) \to 0$ and similarly as $\Gamma_n \subseteq \ty_n$, $\Gamma=\Gamma^r$.
Therefore, we apply Theorem \ref{th:limSan} to obtain
\[\lim_n \frac{1}{n} \log \prob_{\mu}(\ty_{Y^n}\in \Gamma_n) = -\inf_{\nu \in \Gamma} D(\nu\|\mu)\]
or equivalently
\[ \lim_n \frac{1}{n} \log \frac{|C_n|}{|S|^n} = -\inf_{\nu \in \Gamma} (\log|S| - H_\nu(X+Y)) = \sup_{\nu \in \Gamma} H_\nu(X+Y) -  \log|S|.\]
Since $\Gamma$ and $\hat{\Gamma}$ is compact, and 
$ \sup_{\nu \in \Gamma} H_\nu(X+Y) =  \max_{q\in \Pi(p_X,p_Y)} H_q(X+Y)$, we are done.
\end{proof}
\end{appendices}
\end{document}